\documentclass[a4paper,abstract=true,headings=normal]{scrartcl}
\setkomafont{author}{\large}
\usepackage{authblk}

\usepackage{tikz}
\usepackage[dvipsnames]{xcolor}
\usepackage[colorlinks,linkcolor=RedViolet,urlcolor=NavyBlue,citecolor=ForestGreen]{hyperref}
\usepackage{latexsym}
\usepackage{amsmath}
\usepackage{amsthm}
\usepackage{amssymb}
\usepackage{mathtools}
\usepackage{stmaryrd}

\usepackage[backend=bibtex,style=numeric]{biblatex}
\renewbibmacro*{doi+eprint+url}{\iftoggle{bbx:url}
  {\iffieldundef{doi}{\usebibmacro{url+urldate}}{}}
  {}\newunit\newblock
  \iftoggle{bbx:eprint}
  {\usebibmacro{eprint}}
  {}\newunit\newblock
  \iftoggle{bbx:doi}
  {\printfield{doi}}
  {}}

\allowdisplaybreaks 

\newtheoremstyle{mythm}{}{}{\itshape}{}{\bfseries}{.}{.5em}{\thmname{#1}~\thmnumber{#2}\ifthenelse{\equal{\thmnote{#3}}{}}{}{~(\thmnote{#3})}}

\newtheoremstyle{mydefn}{}{}{\upshape}{}{\bfseries}{.}{.5em}{\thmname{#1}~\thmnumber{#2}\ifthenelse{\equal{\thmnote{#3}}{}}{}{~(\thmnote{#3})}}

\newtheoremstyle{myremark}{}{}{\upshape}{}{\itshape}{.}{.5em}{\thmname{#1}~\thmnumber{#2}\ifthenelse{\equal{\thmnote{#3}}{}}{}{~(\thmnote{#3})}}

\theoremstyle{mythm}
\newtheorem{theorem}{Theorem}[section]
\newtheorem{lemma}[theorem]{Lemma}
\newtheorem{proposition}[theorem]{Proposition}
\newtheorem{corollary}[theorem]{Corollary}

\newtheorem{conjecture}[theorem]{Conjecture}

\theoremstyle{mydefn}
\newtheorem{definition}[theorem]{Definition}
\newtheorem{example}[theorem]{Example}

\theoremstyle{myremark}
\newtheorem{remark}[theorem]{Remark}
\theoremstyle{mythm}

\numberwithin{equation}{section}

\newcommand{\pol}{P}

\newcommand{\PTIME}{\text{\upshape P}}
\newcommand{\NP}{\text{\upshape NP}}
\newcommand{\coNP}{\text{\upshape co-NP}}

\DeclareMathOperator{\ifp}{\mathsf{ifp}}
\DeclareMathOperator{\free}{free}
\DeclareMathOperator{\ext}{ext}
\DeclareMathOperator{\ints}{int}

\DeclareMathOperator*{\avg}{avg}
\DeclareMathOperator*{\uniq}{uniq}
\newcommand{\wt}w
\newcommand{\bias}b
\DeclareMathOperator{\dom}{dom}
\DeclareMathOperator{\rg}{rg}
\DeclareMathOperator{\BP}{BP}

\newcommand{\Nat}{\mathbb N}
\newcommand{\PNat}{{\mathbb N}_{>0}}
\renewcommand{\vec}[1]{\boldsymbol{#1}}
\newcommand{\ite}[3]{\operatorname{\mathsf{if}}#1\operatorname{\mathsf{then}}#2\operatorname{\mathsf{else}}#3}
\newcommand{\CF}{{\mathcal F}}
\newcommand{\bigmid}{\mathrel{\big|}}

\newcommand{\class}[1]{{\operatorname{eqcl}(#1)}}
\newcommand{\ind}[1]{{\operatorname{ind}(#1)}}

\newcommand{\phiWUpdate}{\varphi_{\text{weight-update}}}
\newcommand{\phiStop}{\varphi_{\text{stop}}}
\newcommand{\phiAll}{\varphi_{\text{all}}}
\newcommand{\phiAllNew}{\varphi_{\text{all}}^\text{new}}
\newcommand{\phiAllOldNew}{\varphi_{\text{all}}^\text{old-new}}
\newcommand{\phiAllOldOldNew}{\varphi_{\text{all}}^\text{old-old-new}}
\newcommand{\RAll}{R_{\text{all}}}
\newcommand{\RAllOld}{R_{\text{all}}^{\text{old}}}
\newcommand{\RAllOldOld}{R_{\text{all}}^{\text{old-old}}}
\newcommand{\lti}[1]{\varphi_{\text{lti}}^{(#1)}}

\newcommand{\IfText}{\textsf{if }}
\newcommand{\ThenTextNS}{\textsf{then }}
\newcommand{\ThenText}{\text{ }\ThenTextNS}
\newcommand{\ElseTextNS}{\textsf{else }}
\newcommand{\ElseText}{\text{ }\ElseTextNS}

\newcommand{\superN}{^{\Net}}
\newcommand{\f}f \newcommand\Rbf{\mathbf R}
\newcommand\Rlin{\mathbf R_{\text{\upshape lin}}}

\newcommand\bblm[1]{\mathsf{FO}(\Rbf,\f/#1)}
\newcommand\bblinm[1]{\mathsf{FO}(\Rlin,\f/#1)}
\newcommand\bbl{\mathsf{FO}(\Rbf,\f)}
\newcommand\bblin{\mathsf{FO}(\Rlin,\f)}
\newcommand{\R}{\mathbb{R}}
\newcommand{\Q}{\mathbb{Q}}
\newcommand{\Rbot}{{\R_\bot}}
\newcommand{\Qbot}{\Q_{\bot}}
\newcommand{\A}{\mathcal A}
\newcommand{\B}{\mathcal B}
\newcommand{\C}{\mathcal C}

\newcommand{\K}{\mathbf K}
\newcommand{\ass}{\nu}
\newcommand{\sem}[2]{\llbracket #1 \rrbracket^{#2}}
\newcommand{\ar}{\mathrm{ar}}
\newcommand{\fosum}{\textsf{\upshape FO(SUM)}}
\newcommand{\ifpsum}{\textsf{\upshape IFP(SUM)}}
\newcommand{\sifpsum}{\textsf{\upshape sIFP(SUM)}}
\newcommand{\IFP}{\textsf{\upshape IFP}}
\newcommand{\Net}{\mathcal N}
\newcommand{\RNet}{\widetilde{\Net}}
\newcommand{\bvec}[1]{{\boldsymbol #1}}
\newcommand{\xx}{\bvec x}
\newcommand{\yy}{\bvec y}
\newcommand{\zz}{\bvec z}
\newcommand{\ReLU}{\mathrm{ReLU}}

\newcommand{\clipped}{\mathrm{lsig}}
\newcommand{\Netsplit}{\Net^{\mathsf{split}}}
\newcommand{\inn}{\mathrm{in}}
\newcommand{\out}{\mathrm{out}}
\newcommand{\In}{\mathrm{In}}
\newcommand{\Out}{\mathrm{Out}}
\newcommand{\val}{\mathit{val}}

\begin{document}

\title{Recursive querying of neural networks via weighted structures}
\author[1]{Martin Grohe}
\author[1]{Christoph Standke}
\author[2]{Juno Steegmans}
\author[2]{Jan~Van~den~Bussche}
\affil[1]{RWTH Aachen University}
\affil[2]{Hasselt University}
\date{}
\maketitle

\begin{abstract}

Expressive querying of machine learning models---viewed as a form
  of intensional data---enables their verification and
  interpretation using declarative languages, thus making learned
  representations of data more accessible. Motivated by the
  querying of feedforward neural networks, we investigate
  logics for weighted structures. In the absence of a bound on
  neural network depth, such logics must incorporate recursion;
  thereto we revisit the functional fixpoint mechanism proposed
  by Gr\"adel and Gurevich.  We adopt it in a Datalog-like
  syntax; we extend normal forms for fixpoint logics to weighted
  structures; and show an equivalent ``loose'' fixpoint mechanism
  that allows values of inductively defined weight functions to
  be overwritten.  We propose a ``scalar'' restriction of
  functional fixpoint logic, of polynomial-time data complexity,
  and show it can express all PTIME model-agnostic queries over
  reduced networks with polynomially bounded weights.  In contrast, we show
  that very simple model-agnostic queries are already
  NP-complete.  Finally, we consider transformations of weighted
  structures by iterated transductions.

\end{abstract}

\section{Introduction}

A case can be made, from several perspectives, for the
\emph{querying} of machine learning models:
\begin{itemize}
  \item
    Data science projects generate a large amount of
    model artefacts, which should be managed using database technology,
    just like any other kind of data.  In particular, we should
    be able to query this data.  Platforms like MLflow or W\&B
    offer administrative
    filtering and search based on experimental metadata, but
    no deep querying of the models themselves.
  \item

    In machine learning terminology, ``querying'' a model often just means to
    apply it to a new input.  However, we can be much more
    ambitious.  Consider a typical Boolean classifier on tuples
    (vectors) of numeric features.  Such a model represents the
    potentially infinite relation consisting of all possible
    tuples that are classified as true.  We would like to be able to
    query such relations just like ordinary relations in a
    relational database.
  \item
    Querying infinite relations that are finitely presented by
    constraints was already intensively investigated in
    database theory under the heading of \emph{constraint
    query languages} \cite{kkr_cql,cdbbook,libkin_embedded}.
    
  \item
    The multitude of methods for model interpretability or explainable AI
    \cite{molnar-book,rudin-stop-explaining} can be viewed as
    many different queries on models and data (e.g., finding
    a counterfactual, computing the Shapley value), but outside
    the framework of an encompassing structured query language.

  \item
    In the verification of neural networks
    \cite{aws-book,vnn-comp-2024}, models 
    represent functions over the reals, and
    properties to be verified
    are expressed as universally quantified constraint
    expressions about such functions.  These expressions can
    already be thought of as a minimal query language.
\end{itemize}

The above considerations can motivate us to investigate the
theoretical foundations of query languages for neural networks.
Indeed, research in this direction has already been started.
Arenas et al.~\cite{arenas-foil,arenas-dtfoil} consider boolean
decision trees and first-order logics over arbitrary-length
boolean vectors.  They combine these logics with logics that
quantify over the nodes of the decision tree, enabling query
evaluation via SAT solving.  Grohe et al.~\cite{ql4nn} consider
feedforward neural networks (FNNs) and the first-order
logic \fosum\ for weighted structures, with query evaluation via
SQL \cite{sql4nn}.

Focusing on model-agnostic queries,\footnote{A query is
model-agnostic if it does not distinguish between two models that
may be structurally different, but that happen to represent the
same function or classifier \cite{lime-explaining}.} Grohe et
al.\ show a remarkable dichotomy. Without a bound on the depth of
the network, \fosum, due to its lack of recursion, can only
express trivial model-agnostic queries. On fixed-depth FNNs,
however, \fosum\ has substantial expressive power and can express all queries
in the linear-arithmetic fragment of the constraint query
language $\bbl$.  Here, $\bbl$ stands for
first-order logic over the reals with an extra function symbol
$f$.  This language serves as a natural yardstick for
expressing model-agnostic queries: the model is accessed
as a black box via the symbol $f$.

In this paper, we investigate methods for, as well as obstacles
to, lifting the fixed-depth assumption made in Grohe et al.'s
work.  We offer the following contributions.  We begin by
developing the necessary theory for extending \fosum\ with
recursion.  At its core, \fosum\ is a logic for defining weight
functions.  Fixpoint logics for inductive definitions of
\emph{relations} are well known from logic, database theory and
finite model theory.  In contrast, for weight functions, we are
aware of only one proposal, the functional fixpoint
\cite{gg_metafinite}, which we revisit and develop more
systematically.  We examine alternative semantics and show that
they are equally expressive.  We define the recursive extension
of \fosum, called \ifpsum, both in a Datalog-like and in a
fixpoint-logic-like syntax, and establish a normal form that
extends the known normal form for fixpoint logic and inflationary
Datalog with negation \cite{av_datalog,ef_fmt2,libkin_fmt}.

Our treatment of \ifpsum\ is generally valid for all weighted
structures.  In order to apply it to FNNs, we first define a
restricted fragment, called \emph{scalar} \ifpsum, disallowing
multiplication in combination with recursion.  We prove that this
fragment has PTIME data complexity.  Due to the combination of
recursion with arithmetic on weights, this is a nontrivial
exercise.  We then establish that scalar \ifpsum\ can express
\emph{all} PTIME model-agnostic queries, on FNNs that have
polynomially bounded ``reduced'' weights.  We make the latter
condition precise and explain why it is necessary.

Next, we relate back to the black-box logic $\bbl$.  The relevant
question is whether \ifpsum\ can now express all of
linear $\bbl$, without the fixed-depth assumption.
An affirmative answer is unlikely, however.
We will show that very simple $\bbl$ queries are already \NP-hard.
This provides a negative answer at least for scalar \ifpsum,
unless $\PTIME = \NP$.  The question for unrestricted \ifpsum\
remains open.

Finally, we consider a highly expressive extension of \ifpsum\
that allows transductions, mapping structures to structures, to
be defined iteratively.  Reviewing the proof of the
expressiveness result of Grohe et al., we conclude that every
linear $\bbl$ query can indeed be computed by an iteration of
\ifpsum\ transductions.

This paper is organised as follows.  Section~\ref{secrel}
discusses related work. Section~\ref{secprel} presents
preliminaries. Section~\ref{secifp} contains the results on
\ifpsum, and Section~\ref{secscalar} those on the scalar
fragment. Section~\ref{secnph} presents the \NP-completeness
result.  Section~\ref{seciter} discusses iterative
transductions. Section~\ref{seconc} concludes.  An appendix with
proof details is provided.

\section{Related work} \label{secrel}

The logic \fosum\ for weighted structures is an instantiation of
the logics for metafinite structures defined by Gr\"adel and
Gurevich \cite{gg_metafinite}.  Apparently unaware of that work,
Torunczyk made a very similar proposal and studied the combined
complexity of query evaluation and enumeration over general
classes of numerical domains (semirings) \cite{toruncz-faq}.  In
a statistical-learning context, Van Bergerem and
Schweikardt~\cite{BergeremS21,BergeremS25} study a closely
related extension of first-order logic by "weight aggregation".
Of course, \fosum\ is also quite similar to the relational
calculus with aggregates and arithmetic which formalises basic
SQL \cite{libkin_sql}.  The difference is that in SQL, relations
can have multiple numerical columns, while metafinite structures
have a separate abstract domain on which relations and weight
functions (taking values in a separate numerical domain) are
defined.

Weight functions are also very similar to semiring-annotated
relations \cite{provsemirings,faq_sigmodrecord}. However, in that
space, the focus is typically on languages where the weights are
\emph{implicitly} added, multiplied, or summed via corresponding
relational operators. In contrast, \fosum\ deals
\emph{explicitly} with the numerics.\footnote{See also work on
implicit versus explicit handling of nonnumeric annotations
\cite{stijn_provenance,color_journal}.}  Explicit tensor logics
were also proposed by Geerts and Reutter
\cite{gr-gnn-tl,floris-qlgltut} for the purpose of characterising
indistinguishability of weight functions by graph neural
networks.  In contrast, our focus here is on 
model-agnostic querying of feedforward neural networks.

To this aim, we add recursion to \fosum.
When mixing recursion and numerical computation, termination
becomes a point of attention.  We use and
investigate variations of the inflationary fixpoint that are
guaranteed to terminate.  In contrast, in very interesting
related work, Abo Khamis et al.~consider
datalog over semiring-annotated relations and investigate conditions 
on the semiring that guarantee convergence \cite{datalogo}.
We also mention work on the semantics for logic programs with
aggregates \cite{necker-aggregates}.  There, the focus is on
providing natural semantics for general sets of recursive rules
involving negation and aggregation, and finding characterisations
for when such rule sets are unambiguous (have unique well-founded
models).

Work related to our \NP-hardness result in Section~\ref{secnph}
has considered the complexity of deciding various properties of
FNNs, including injectivity and surjectivity of the represented
function \cite{FroeseGS25}.  Also, various other forms of neural
network verification have been shown to be \NP-complete
\cite{Wurm24}.  However, existing results crucially depend on the
assumption that the \emph{width} (input dimension) of the network
to be verified contributes to the input size of the verification
problem.  In contrast, our \NP-hardness result is stronger in that
it already pertains to networks of width one.

\section{Preliminaries} \label{secprel}

\subsection{FNNs} \label{secfnn}

The structure of a feedforward neural network
\cite{understanding-ml-book}, abbreviated FNN, is that of a
directed acyclic graph with weights on nodes and edges.  The
source nodes are called \emph{inputs} and are numbered from $1$
to $m$; the sink nodes are called \emph{outputs} and are numbered
from $1$ to $p$.
All other nodes are called \emph{hidden} nodes.
We always assume input and output nodes to be
distinct, that is, we disallow a graph of only isolated nodes.
The weight of a node is also called its \emph{bias};
exceptions are the input nodes, which do not have a bias.

Let $\Net$ be an FNN as described.  The function $\f^{\Net}: \R^m
\to \R^p$ represented by $\Net$, using ReLU activations and
linear outputs, is defined as follows.  We begin by defining, for
every node $u$, a function $\f^{\Net}_u : \R^m \to \R$ by
induction on the \emph{depth} of $u$
(the maximum length of a path from an input node to $u$.) If $u$
is the $i$th input node then $\f^{\Net}_u(x_1,\dots,x_m) = x_i$.
If $u$ is a hidden node with bias $b$, and incoming edges
$(v_1,u)$, \dots, $(v_k,u)$ with weights $w_1$, \dots, $w_k$,
respectively, then $\f^{\Net}_u(\xx) = \ReLU(b + \sum_j w_j
\f^{\Net}_{v_j}(\xx))$. Here, $\ReLU : \R \to \R : z \mapsto
\max(0,z)$.  If $u$ is an output node, $\f^\Net_u(\xx)$ is defined
similarly as for hidden nodes, but the application of $\ReLU$ is
omitted.  We can finally define $\f^{\Net}(\xx)$ as
$(\f^{\Net}_{\out_1}(\xx),\dots,\f^{\Net}_{\out_p}(\xx))$, where
$\out_1$, \dots, $\out_p$ are the output nodes.

$\K(m,p)$ denotes
the class of FNNs with $m$ inputs and $p$ outputs.
We also write $\K(*,p)$, $\K(m,*)$, and $\K(*,*)$ when
the numbers of inputs, or outputs, or both, are not fixed.

\subsection{Model-agnostic queries and $\bbl$}
\label{secmodelagnostic}

In general, we may define an \emph{$r$-ary query
on $\K(m,p)$ with $k$ parameters} to be a relation $Q \subseteq
\K(m,p) \times \R^k \times \R^r$.  If $Q(\Net,\bvec z, \bvec y)$
holds, we say $\bvec y$ is a possible result of $Q$ on $\Net$ and
$\bvec z$.  In the special case $r=0$, we obtain a \emph{boolean
query} (true if $Q(\Net,\bvec z)$ holds, false otherwise).

\begin{example} \label{exqueries}
The simplest example of a query is inference, i.e.,
evaluating a model on a given input.  Formally, for inference, we
have $k=m$ and $r=p$ and $Q(\Net,\zz,\yy)$ holds iff $\yy =
\f\superN(\zz)$.  However, many more
tasks in interpretable machine learning \cite{molnar-book} fit
the above notion of query. For example, returning a
  counterfactual explanation or an
adversarial example, checking robustness, computing the
Shapley value in a point, computing the gradient in a point, or
  checking differentiability between certain ranges given by the
  parameters.

The case $k=0$ captures analysing or verifying
the global behaviour of neural networks, instead of their
  behaviour on given input vectors.  For example, recall that
  the function
  represented by an FNN with ReLU and linear outputs is always
  piecewise linear.  A possible $2$-ary query for $m=1=p$ on an
  FNN $\Net$ may ask for results $(b,s)$ such that $b$ is a
  breakpoint of $\f\superN$ and $s$ is the right-hand slope at
  $b$. For higher input dimensions, we could, for
  example, ask for the coefficients of supporting hyperplanes of the
  partitioning of input space induced by $\f\superN$.
  \qed
\end{example}

All examples of queries just mentioned are \emph{model-agnostic},
meaning that the query has the same results on two networks
$\Net$ and $\Net'$ representing the same function, i.e.,
$\f\superN=\f^{\Net'}$.  Model-agnostic methods form an
established category in interpretable machine learning and
explainable AI \cite{molnar-book}.\footnote{To give an example of
a method that is not model-agnostic, we can mention the
pruning of a neural network:
finding neurons that have only a negligible influence on the
function that is represented.}

As a yardstick for model-agnostic queries on $\K(m,p)$, we can
use first-order logic over the reals with $p$ function symbols
$f_1$, \dots, $f_p$ of arity $m$, together representing a
function $f : \R^m \to \R^p$.  For simplicity, in this
paper, we will look mainly at the case $p=1$, which in itself is
already typical in machine learning, with tasks such as regression and
binary classification.

Formally, let $\Rbf=(\R,+,\cdot,(q)_{q\in\Q},<)$ be the structure
of the reals with constants for all rational numbers.  By $\bblm
m$, we mean first-order logic over the vocabulary of $\Rbf$ with
an extra $m$-ary function symbol $f$.  When multiplication is
restricted to be only between a term and a constant (scalar
multiplication), we denote this by $\Rlin$ and $\bblinm m$.
When $m$ is understood, we omit it from the notation.

A formula $\varphi$ of $\bblm m$ with $k+r$ free variables
$z_1,\dots,z_k,y_1,\dots,y_r$ now naturally expresses an $r$-ary query
$Q_\varphi$ on $\K(m,1)$ with $k$ parameters.  Specifically,
$Q_\varphi(\Net,\zz,\yy)$ holds iff $\varphi(\zz,\yy)$ is satisfied in
$(\Rbf,\f\superN)$.  This query is model-agnostic by
definition.

\begin{example}
  We give two formulas to illustrate the syntax of $\bbl$.  The
  inference query (Example~\ref{exqueries})
  is expressed by the formula $y = f(\zz)$.  
  The query on $\K(1,1)$ that asks whether
  $\lim_{x\to+\infty}\f\superN(x) \allowbreak =-\infty$,
  is expressed by $\forall u<0 \, \exists x_0>0 \, \forall x>x_0
  \, f(x) < u$.
  Both formulas are in $\bblin$ since
  they do not use multiplication.
We note that
all queries from Example~\ref{exqueries}, with the
exception of Shapley value, are expressible in $\bbl$ \cite{ql4nn}.
\end{example}

\subsection{\fosum} \label{sec:fosum}

\fosum\ is a logic for querying weighted structures, which are
standard relational structures additionally equipped with weight
functions.  These functions map tuples of elements to numeric
values.

In our case, the numeric values are taken from the ``lifted''
rationals $\Qbot\coloneqq\Q\cup\{\bot\}$ or reals
 $\Rbot = \R \cup \{\bot\}$.
 Here, $\bot$ is an
extra element representing an undefined value.  We extend the
usual order $\leq$ on $\R$ to $\Rbot$ by
letting $\bot \leq x$ for all $x \in \Rbot$. We
extend addition, subtraction, and multiplication by letting $x+y
:= \bot$, $x-y := \bot$, and $x \cdot y := \bot$ if $x = \bot$ or
$y = \bot$. Similarly, we extend division by letting $x/y:=
\bot$ if $x = \bot$ or $y = \bot$ or $y = 0$.

A \emph{weighted vocabulary} $\Upsilon$ is a finite set of
relation symbols and weight-function symbols. Each symbol $S$ has
an arity $\ar(S)$, a natural number.
A \emph{weighted $\Upsilon$-structure} $\A$ consists of a
finite set $A$ called the universe of $\A$, for each $k$-ary
relation symbol $R \in \Upsilon$ a $k$-ary relation
$R^{\A} \subseteq A^k$, and for each $k$-ary
weight-function symbol $F \in \Upsilon$ a function
$F^{\A} : A^k \to \R_{\bot}$.  We will often refer to weighted
structures simply as structures for short.

We define the sets of \emph{formulas} $\varphi$ and \emph{weight
terms} $\theta$ of the logic \fosum\
by the following grammar:

\begin{align}
  \label{eq:fosum-formulas}
\varphi &::= x = y \mid R(x_1, \ldots, x_{\ar(R)}) \mid \theta
  \leq \theta \mid \neg\varphi \mid \varphi \ast \varphi \mid Qx\, \varphi \\
  \label{eq:fosum-terms}
\theta &::= r \mid F(x_1, \ldots, x_{\ar(F)}) \mid \theta \circ
  \theta \mid \textsf{if $\varphi$ then $\theta$ else $\theta$}
  \mid \sum_{(x_1, \ldots, x_k):\varphi}
  \theta.
\end{align}

Here $x, y, x_i$ are variables, $R$ is a relation symbol,
$\ast \in \{\vee, \wedge, \to\}$ is a Boolean connective,
$Q \in \{\exists, \forall\}$ is a quantifier, $r \in \Qbot$ is a
numeric constant,\footnote{One can also allow
arbitrary real constants, but in this paper, we are
  concerned with algorithms evaluating expressions, and therefore
  reasonably restrict our attention to rational constants.}
$F$ is a weight-function
symbol, and $\circ \in \{+, -, \cdot, /\}$ is an arithmetic
operator. The semantics is defined with respect to pairs $(\A, \ass)$,
where $\A$ is a structure and $\ass$ an assignment of elements from
the universe $A$ to the variables. Formulas $\varphi$ take a Boolean
value $\llbracket\varphi\rrbracket^{(\A,\ass)} \in \{0, 1\}$ and
weight terms $\theta$ take a value
$\llbracket\theta\rrbracket^{(\A,\ass)} \in \R_{\bot}$. These values
are defined inductively; we omit most definitions as they are obvious
or follow the familiar semantics of first-order logic. The semantics
of the summation operator is as follows:

\[ \llbracket \sum_{(x_1, \ldots, x_k):\varphi} \theta
\rrbracket^{(\A,\ass)} := \sum_{(a_1, \ldots, a_k) \in A^k}
\llbracket\varphi\rrbracket^{(\A,
\ass\frac{a_1, \ldots, a_k}{x_1, \ldots, x_k})} \cdot
\llbracket\theta\rrbracket^{(\A,
\ass\frac{a_1, \ldots, a_k}{x_1, \ldots, x_k})} , \]
where
$\ass\frac{a_1, \ldots, a_k}{x_1, \ldots, x_k}$
denotes the updated assignment
obtained from $\ass$ by assigning $a_i$ to $x_i$ for
$i=1,\dots,k$.

An \fosum\ \emph{expression} is either a formula or a weight
term. The set $\free(\xi)$ of \emph{free variables} of an expression
$\xi$ is defined in a
straightforward way, where a summation $\sum_{(x_1, \ldots,
x_k):\varphi}$ binds the variables $x_1, \ldots, x_k$. A
\emph{closed expression} is an expression without free variables.
A \emph{closed formula} is also called a \emph{sentence}.

For an
expression $\xi$, the notation $\xi(x_1, \ldots, x_k)$ stipulates
that all free variables of $\xi$ are in $\{x_1, \ldots, x_k\}$.
It is easy to see that the value
$\llbracket\xi\rrbracket^{(\A,\nu)}$ only depends on the values
$a_i := \nu(x_i)$ of the free variables. Thus, we may avoid explicit
reference to the assignment $\nu$ and write
$\llbracket\xi\rrbracket^{\A}(a_1, \ldots, a_k)$ instead of
$\llbracket\xi\rrbracket^{(\A,\nu)}$. If $\xi$ is a closed
expression, we just write $\llbracket\xi\rrbracket^{\A}$. For
formulas $\varphi(x_1, \ldots, x_k)$, we also write $\A \models
\varphi(a_1, \ldots, a_k)$ instead of
$\llbracket\varphi\rrbracket^{\A}(a_1, \ldots, a_k) = 1$, and for
sentences $\varphi$ we write $\A \models \varphi$.

Observe that we can express averages in \fosum\ using summation and
division. It will be useful to introduce a notation for averages: for
a term $\theta$ and formula $\varphi$ we write
$\avg_{\vec x:\varphi}\theta$ to abbreviate
$(\sum_{\vec x:\varphi}\theta)/\sum_{\vec
    x:\varphi}1$. Note that $\avg_{\vec x:\varphi}\theta$
takes value $\bot$ if there are no tuples $\vec x$ satisfying
$\varphi$.

\subsection{FNNs as weighted structures} \label{secfnnstruct}

Any FNN $\Net\in\K(*,*)$ can be naturally regarded as a weighted
structure
\[
  \Net =
  (V,E\superN,\In\superN,\Out\superN,b\superN,\allowbreak w\superN),
\]
where
$V$, the universe, is the set of nodes; $E\superN$ is the binary
edge relation; $\In\superN$ is a binary relation that is a linear
order of the input nodes of $\Net$ (and undefined on the
remaining nodes); $\Out\superN$ similarly is a
linear order of the output nodes; $b\superN$ is the unary bias
weight function on nodes; and $w\superN$ is the binary weight
function on edges.\footnote{To be precise,
$b\superN(u)=\bot$ for every input node $u$, and
$w\superN(u,v)=\bot$ for every pair $(u,v)$ not in $E$.}

This slightly generalises the model of \cite{ql4nn}, where the
vocabulary depended on the input dimension and output dimension of the
network, and for a network $\Net\in\K(p,q)$, the $p$ input nodes were
accessed by a singleton $p$-ary relation and the $q$ output nodes were
accessed by a singleton $q$-ary relation. Note that we can easily
retrieve these relations in our version. The singleton input relation
of an FNN in $\K(p,q)$ can be defined by the $\fosum$ formula
$
\varphi_{\textup{In}(p)}(x_1,\ldots,x_p)\coloneqq\bigwedge_{i=1}^{p-1}
\big(\In(x_i,x_{i+1})\wedge x_i\neq x_{i+1}\big)\wedge\forall
y\left(\In(y,y)\to\bigvee_{i=1}^p y=x_i\right), $ and similarly for
the output relation.  The advantage of our approach is that
we can write queries that apply uniformly to all FNNs, regardless of
their dimension.  

We can expand the corresponding vocabulary $(E,\In,\Out,b,w)$
with a weight function $\val$ giving
input values to the network.  Over the resulting
vocabulary, we can now give a few examples of formulas and
weight terms in \fosum.

\newcommand{\famdepth}[1]{\mathit{depth}_{\leq #1}}
\newcommand{\tameval}[1]{\mathit{eval}_{\leq #1}}

\begin{example} \label{exfixedepth}
For any natural number $\ell$, there is a
first-order logic formula
  $\famdepth\ell(u)$ defining the nodes of depth at most $\ell$
  in any network.  We can then define the function $\f\superN_u$
(cf.~Section~\ref{secfnn}) for such nodes
by the weight term $\tameval\ell(u)$ defined as
  follows:
\begin{tabbing}
  $\tameval 0 := \textsf{if $\In(u,u)$ then $\val(u)$ else $\bot$}$ \\
  $\tameval {\ell+1} := {}$\=\textsf{if $\famdepth \ell(u)$ then
  $\tameval \ell(u)$}
  \textsf{else if $\famdepth{\ell+1}(u)$ then} \+\\
  \hspace*{1em}\textsf{if $\Out(u,u)$ then
  $b(u) + \sum_{v : E(v,u)} w(v,u) \cdot \tameval\ell(v)$} \\
  \hspace*{1em}\textsf{else}
  $\ReLU(b(u) + \sum_{v : E(v,u)} w(v,u) \cdot \tameval\ell(v))$ \\
  \textsf{else $\bot$}
\end{tabbing}
Here, $\ReLU(\theta)$, for an arbitrary weight term $\theta$, is
an abbreviation for \textsf{if $\theta \geq 0$ then $\theta$ else
$0$}. \qed
\end{example}

The above example essentially shows that \fosum\ can express the
inference query on fixed-depth networks
(cf.~Section~\ref{secmodelagnostic}). 
It turns out that \fosum\ can do much more than that and actually
measures up to the yardstick set by $\bblin$:
\begin{theorem}[\cite{ql4nn}] \label{theorfixedepth}
  Let $m$, $k$ and $\ell$ be natural numbers, and let $Q$ be a boolean
  query on $\K(m,1)$ with $k$ parameters, expressible in $\bblin$. 
  There exists an \fosum\ sentence $\psi$ over vocabulary
  $(E,\In,\Out,b,w,\val)$ that expresses $Q$ on
  all networks in $\K(m,1)$ of depth $\ell$.
\end{theorem}
This result shows that \fosum\ can simulate
arbitrary quantification over the real numbers,
a feature central to $\bblin$. This is remarkable since
\fosum\ itself
only offers quantification over the finite set of nodes of the network.

\section{Extending \fosum\ with recursion} \label{secifp}

This paper investigates methods for, as well as obstacles to,
lifting the fixed-depth restriction of
Theorem~\ref{theorfixedepth}.  A necessary development, which we
investigate in this section, is to extend \fosum\ with a
recursion mechanism, allowing for inductive definitions.

Recall that \fosum\ expressions can be formulas, which define
relations, or weight terms, which define weight functions.  For
inductive definitions of relations, logical mechanisms are well
understood; in this paper, we consider a Datalog-like syntax with
stratification and the inflationary fixpoint semantics for strata
\cite{av_datalog,kp-whynot,ahv_book}.  We recall this semantics through an
example.

\newcommand{\Ans}{\mathit{Ans}}

\begin{example} \label{exans}
  Consider finite directed graphs given by a binary edge
  relation $E$ with an extra unary relation $S$.  The following
  program checks that the graph is acyclic when restricted to the
  nodes reachable from nodes in $S$:
  \newcommand{\Reach}{\mathit{Reach}}
  \begin{tabbing}
    $\Reach (x) \gets S(x) \lor \exists y(Reach(y) \land
    E(y,x))$; \\
    $\Ans \gets \neg \exists x\, Reach(x,x)$.
  \end{tabbing}
  The relation $\Reach$ is initialised to be empty. The defining
  formula is evaluated repeatedly, and the result is \emph{added} to
  $\Reach$ until no change occurs.\footnote{The adding gives the
  inflationary aspect; the defining formula need not be positive
  or monotone.}  Nullary relation $\Ans$
  provides the boolean answer to the query and is defined
  non-recursively.  Its defining rule is in a subsequent stratum,
  meaning that it
is evaluated after the recursion for $\Reach$ has terminated. \qed
\end{example}

To our knowledge,
inductive definitions of \emph{weight functions} has been much less
studied.\footnote{In standard first-order logic, whose syntax lacks the
powerful interaction between formulas and weight terms we have in
\fosum, one is limited to defining functions inductively as relations (graph
of a function) using formulas,
and somehow declaring the definition to be wrong
if the resulting relation is not the graph of a function.}
We can adopt the only existing proposal, which is
nicely compatible with inflationary fixpoints, called
\emph{functional fixpoints} \cite{gg_metafinite}.  The idea is
that an inductively defined weight function $F$ is initialised to be
undefined ($\bot$) everywhere.  A defining weight term is then
evaluated repeatedly and used to update $F$, but only on tuples where $F$
was not yet defined.  This is repeated until no change occurs.

\newcommand{\feval}{\mathit{eval}}

\begin{example} \label{ex:receval}
  The following inductively defined function $\feval$ uniformly
  expresses, given any FNN $\Net$, the function $\f\superN_u$ for
  all nodes. (Compare the weight term $\tameval \ell$ from
  Example~\ref{exfixedepth}, which does the same but only for nodes $u$ up
  to depth $\ell$.)
  \begin{tabbing}
    $\feval(u) \gets {}$\=\textsf{if $\In(u,u)$ then $\val(u)$} \\
  \>\textsf{else if
    $\Out(u,u)$ then $b(u) + \sum_{v : E(v,u)} w(v,u)\cdot
    \feval(v)$} \\
    \>\textsf{else $\ReLU(b(u) + \sum_{v : E(v,u)} w(v,u)\cdot\feval(v))$}.
  \end{tabbing}
\end{example}

\subsection{The logic \ifpsum}
\label{sec:ifpsum}

We formally define \ifpsum: the extension of \fosum\ with inflationary and
functional fixpoints.  Our syntax follows the pattern of
stratified Datalog \cite{ahv_book}.  Later in this section, we will
also consider a syntax that is more in line with the way fixed-point
extensions of first-order logic are defined in finite model
theory \cite{ef_fmt2,libkin_fmt}.

\subsubsection*{Rules and strata}

A \emph{relational rule} over a weighted vocabulary
    $\Upsilon$ is of the form $R(\bvec
    x) \gets \varphi$, where $R \in \Upsilon$ is a
    relation name, $\bvec x$ is a tuple of $\ar(R)$ distinct
    variables, and $\varphi(\bvec x)$ is an \fosum\
    formula over $\Upsilon$.

A \emph{weight function rule} over $\Upsilon$ is of the form $F(\bvec
    x) \gets \theta$, where $F \in \Upsilon$ is a
    weight function name, $\bvec x$ is a tuple of $\ar(R)$ distinct
    variables, and $\theta(\bvec x)$ is a weight term over $\Upsilon$.

    Let $\Upsilon$ and $\Gamma$ be two disjoint weighted
    vocabularies.  An \emph{\ifpsum\ stratum of type $\Upsilon \to \Gamma$}
    is a set $\Sigma$ of rules over $\Upsilon \cup \Gamma$, with one rule for
    each symbol in $\Gamma$.  In the spirit of stratified Datalog
    terminology, in $\Sigma$, the symbols from $\Upsilon$
    are called \emph{extensional} and those from $\Gamma$
    \emph{intensional}.

To define the semantics of an \ifpsum\ stratum
$\Sigma$ as above we first introduce the \emph{immediate
consequence} $T_\Sigma$ on $\Upsilon\cup\Gamma$-structures.

\begin{definition} \label{deft}
Given structure $\B$, we define $T_\Sigma(\B) := \C$, where the
universe of $\C$ equals $B$, the universe of $\B$, and
$\C$ agrees with $\B$ on the extensional
symbols; the intensional symbols in $\C$ are then defined as follows.
\begin{itemize}
  \item
    Let $R(\bvec x) \gets \varphi$ be a relational rule in
    $\Sigma$.  Then $R^{\C} := R^\B \cup \{\bvec a \in B^{\ar(R)}
    \mid \B \models \varphi(\bvec a)\}$.
  \item
    Let $F(\bvec x) \gets \theta$ be a weight function rule in
    $\Sigma$.  Then
    \begin{equation}
      \label{eqstrict}
     F^{\C}(\bvec a) := \begin{cases}
      \sem\theta\B(\bvec a) & \text{if $F^\B(\bvec a) = \bot$;} \\
      F^\B(\bvec a) & \text{otherwise}.
    \end{cases}
    \end{equation}
\end{itemize}
\end{definition}

Since $\B$ is finite and intensional relations and weight
functions only grow under the application of $T_\Sigma$, there exists a
natural number $n$ such that $T_\Sigma^n(\B)$ (i.e., the result
of $n$ successive applications of $T_\Sigma$ starting from $\B$)
equals $T_\Sigma^{n+1}(\B)$.  We denote this result by $T_\Sigma^\infty(\B)$.

We are now ready to define the semantics of $\Sigma$ as a
mapping from $\Upsilon$-structures to
$\Upsilon\cup\Gamma$-structures.  Let $\A$ be an
$\Upsilon$-structure and let $\A'$ be its expansion to an
$\Upsilon\cup\Gamma$-structure by setting all intensional
relations to empty and all intensional weight functions to be
undefined everywhere.  Then $\Sigma(\A) := T_\Sigma^\infty(\A')$.

\subsubsection*{Programs and queries} An \ifpsum\ program is just a finite
sequence of strata.  Formally, every stratum, of type
$\Upsilon \to \Gamma$, is also a program of that type;
moreover, if $\Pi$ is a program of type $\Upsilon \to
\Gamma_1$ and $\Sigma$ is a stratum of type $\Upsilon \cup
\Gamma_1 \to \Gamma_2$, then $\Pi ; \Sigma$ is a program of type
$\Upsilon \to \Gamma_1 \cup \Gamma_2$.  The semantics is given
simply by sequential composition.

As illustrated in Example~\ref{exans},
it is customary to designate an \emph{answer symbol} (relation name or
weight function name) from among the intensional symbols of a
program.  In this way, we can use programs to express
\emph{queries}, i.e., mappings from structures to relations or
weight functions.

\subsection{A loose semantics}

The immediate consequence operator just defined
only allows to set a new value for a weight function $F$ on a
tuple $\bvec a$ if $F$ was not yet defined on $\bvec a$.  The
advantage of the resulting functional fixpoint semantics
is that it is guaranteed to terminate on finite
structures.  In practice, however, it may be convenient to be
able to perform updates on weight functions.

\begin{example} \label{floydwarshall}
  \newcommand{\chosen}{\mathit{chosen}}
  \newcommand{\ord}{\mathit{ord}}
  \newcommand{\nxt}{\mathit{next}}
  \newcommand{\last}{\mathit{last}}
  Consider a distance matrix $W$,
represented as a structure
  with a strict total order relation $\ord$ and
  a binary weight function $W$.
  An almost literal transcription of the 
Floyd-Warshall algorithm for all-pairs
  shortest-path distances in \ifpsum\ presents itself below.
  \begin{align*}
    \chosen(k)\gets\;&\nxt(k);
    \\
    D(i,j)\gets\;
                     &\textsf{if $\chosen=\emptyset$ then}\\
                     &\hspace{2em}\textsf{if $\exists k(\nxt(k) \land W(i,j) >
                       W(i,k)+W(k,j))$ then}\\
                     &\hspace{4em}\sum_{k:\nxt(k)} W(i,k)+W(k,j)\\
                     &\hspace{2em}\textsf{else $W(i,j)$} \\
                     &\textsf{else if $\exists k(\nxt(k) \land D(i,j) >
                D(i,k)+D(k,j))$ then}\\
                     &\hspace{4em}\sum_{k:\nxt(k)} D(i,k)+D(k,j) \\
                     &\hspace{2em}\textsf{else }D(i,j).
  \end{align*}
  Here, $i$, $j$ and $k$ are
  variables, and $\nxt(k)$
  abbreviates $\forall x(\ord(x,k)
  \leftrightarrow \chosen(x))$.
  The auxiliary intensional relation $\chosen$ is used to
  iterate over all nodes in the graph; $\nxt(k)$ selects the
  next node.

  Under the functional fixpoint semantics we are using so far, however,
  this program does not work as intended.  After the first
  iteration, $D$ is already everywhere defined and further
  updates to $D$ will not be made. The program would work as
  intended under a more permissive
  semantics that allows weight functions to be updated.  \qed
\end{example}

The above example suggests an alternative \emph{loose fixpoint} semantics for
strata, where updates to weight functions are possible.  Care must be
taken, however, since termination is then no longer
guaranteed.  A simple solution we propose is to stop when a
fixpoint is reached on the intensional \emph{relations} only.

To define the loose semantics formally for a stratum $\Sigma$ of
type $\Upsilon \to \Gamma$, we introduce an alternative immediate
consequence operator $L_\Sigma$ (using $L$ for `loose').
It is defined like $T_\Sigma$ (Definition~\ref{deft}),
except that Equation~\eqref{eqstrict} is replaced simply by
$F^\C(\bvec a) := \sem\theta\B(\bvec a)$.
The definition for intensional relations is not changed.
Since we work with finite structures $\B$
and intensional relations can only grow, there exists a natural
number $n$ such that $L_\Sigma^n(\B)$ and 
$L_\Sigma^{n+1}(\B)$ agree on the intensional relations.
Taking $n$ to be the smallest such number, we define
$L_\Sigma^\infty(\B) := L_\Sigma^n(\B)$
and call $n$ the \emph{loose termination index}.
Note that the loose semantics is only useful
if there are some intensional relations, for otherwise this index is
zero.

\newcommand{\loose}{^L}
The result of applying a stratum $\Sigma$ to an
$\Upsilon$-structure $\A$, now under the loose semantics,
is defined as before, but using
$L_\Sigma^\infty$ instead of $T_\Sigma^\infty$, and denoted by
$\Sigma\loose(\A)$.  The result of a program (sequence of
strata) $\Pi$ on $\A$, using loose semantics for the strata, is
denoted by $\Pi\loose(\A)$.

\subsection{Comparing the two semantics}
It is not difficult to see that
the functional fixpoint semantics can be simulated in the loose
semantics:

\begin{proposition} \label{func2loose}

  For every stratum $\Sigma$ of type $\Upsilon \to \Gamma$ there
  exists a stratum $\Sigma'$ of type $\Upsilon \to \Gamma'$, with
  $\Gamma \subseteq \Gamma'$, such that $\Sigma'^L(\A)$
  and $\Sigma(\A)$ agree on $\Gamma$.

\end{proposition}
\begin{proof}
  For the relation symbols in $\Gamma$, we have the same rule in
  $\Sigma'$ as in $\Sigma$.  Each weight function rule $F(\xx)
  \gets \theta$ in $\Sigma$ is replaced in $\Sigma'$ by 
  $F(\xx) \gets \textsf{if $F(\xx)=\bot$ then $\theta$ else
  $F(\xx)$}$.  Finally, to capture the functional fixpoint, 
  for each $F$ as above we include an extra
relation symbol $R_F$ in $\Gamma'$.
  The rule for $R_F$ is $R_F(\xx) \gets \theta\neq\bot$, thus
  keeping track of the tuples on which $F$ is defined. 
\end{proof}

Conversely, the functional fixpoint semantics can simulate the
loose semantics.  We can prove this by applying \emph{timestamping}
and \emph{delayed evaluation} techniques developed for
inflationary datalog with negation \cite{vianu-unchained}.  Here,
a timestamp of a stage in the loose fixpoint is a concatenation
of tuples, one for each intensional relation, so that at least
one of the tuples is newly added to its intensional relation.
Such timestamps become extra arguments of the intensional weight
functions.  We can then simulate function updates by defining the
function on new timestamps.  We can maintain a weak order on the
timestamps, or use an extra set of ``delayed'' intensional
relation names, so we can identify the timestamps from the
previous iteration as well as the current one.  Upon termination
of the simulating stratum, a subsequent stratum can be used to
project on the last timestamps to obtain the final result of the
loose fixpoint.

We conclude:

\begin{theorem} \label{theor-loose}
  \ifpsum\ programs
  under the functional fixpoint semantics and\/ \ifpsum\ programs
  under the loose fixpoint semantics are equivalent query
  languages for weighted structures.
\end{theorem}

\begin{example} \label{exloose2func}
  \newcommand{\chosen}{\mathit{chosen}}
  \newcommand{\ord}{\mathit{ord}}
  \newcommand{\last}{\mathit{last}}
  \newcommand{\nxt}{\mathit{next}}
  The program from Example~\ref{floydwarshall} under the loose
  fixpoint semantics is simulated by the following program under
  the functional fixpoint semantics:
  \begin{align*}
    \chosen(k') \gets\;
    & \nxt(k'); \\
    D'(k',i,j) \gets\;
    &\textsf{if $\neg \nxt(k')$ then $\bot$} \\
    &\textsf{else if $\chosen=\emptyset$ then} \\
    &\hspace{4em}\textsf{if $W(i,j) > W(i,k')+W(k',j))$ then
      $W(i,k')+W(k',j)$}\\
    &\hspace{4em}\textsf{else $W(i,j)$} \\
    &\hspace{2em}\textsf{else if $\exists k(\last(k) \land D'(k,i,j) >
      D'(k,i,k')+D'(k,k',j))$ then}\\
    &\hspace{6em}\sum_{k:\last(k)}
      D'(k,i,k')+D'(k,k',j) \\
    &\hspace{4em}\textsf{else } \sum_{k:\last(k)} D'(k,i,j);\\
    D(i,j) \gets\;&\sum_{k:\last(k)} D'(k,i,j).
  \end{align*}
Here, $\last(k)$ is an abbreviation for $\chosen(x) \land
  \forall x((\chosen(x) \land x\neq k)\to \ord(x,k))$, which
  selects the node chosen in the previous iteration.
  The final rule defining $D$ is in a separate stratum.
  \qed
\end{example}

We remark that the above example is much simpler than the general
  construction in the proof of Theorem~\ref{theor-loose}. Indeed,
  in the example, we may assume a total order, getting timestamps
  for free.  The theorem, however,
  holds in general.
    It is an open question
whether timestamping is necessary for going from loose to functional
fixpoints.  Specifically, does there exist a query from weighted
structures to weight functions that is expressible in \ifpsum\ using
only unary intensional weight functions under the loose semantics, but
not under the functional fixpoint semantics?

\subsection{A Normal Form}

In this section, we consider \ifpsum\ in the framework and
language of classical fixed-point logics, as studied in recursion
theory and finite model theory.  We will prove a normal form
essentially stating that every \ifpsum-program is equivalent to a
program consisting of a single stratum with a single intensional
symbol followed by a selection and projection.

The redefinition of \ifpsum\ as a logic extending \fosum\
is standard \cite{moschovakis,ef_fmt2,gg_metafinite}, so we will be brief.
We add a new weight term formation rule to the grammar
\eqref{eq:fosum-formulas}--\eqref{eq:fosum-terms}:
\begin{equation}
  \label{eq:ifpsum-terms}
  \theta::=\ifp\big(
  F(x_1,\ldots,x_{\ar(F)})\gets\theta\big)(x'_1,\ldots,x'_{\ar(F)}),
\end{equation}
where $F$ is a weight function symbol; note that $\theta$ can
contain inner $\ifp$ operators.

The free variables of such an $\ifp$ term are $(\free(\theta) -
\{x_1,\ldots,x_{\ar(F)}\} \cup \{x'_1,\ldots,x'_{\ar(F)}\}$.  The
free, or \emph{extensional}, relation and function symbols in an
\ifpsum-expression $\xi$, denoted by $\ext(\xi)$, are defined in
a straightforward way, letting the $\ext$ of $\ifp$-term
\eqref{eq:ifpsum-terms} be $\ext(\theta)-\{F\}$.  A symbol $F$ is
\emph{intensional} in $\xi$ if it appears in a subterm $\ifp\big(
F(\xx)\gets\theta\big)(\xx')$ of $\xi$. We denote the set of all
intensional symbols of $\xi$ by $\ints(\xi)$. \emph{In the
following, for all\/ \ifpsum\ expressions $\xi$ we assume that
$\ints(\xi)\cap\ext\xi)=\emptyset$ and that every intentional
symbol is only bound by a single $\ifp$ operator.} This is
without loss of generality by renaming intensional symbols.

The semantics of \ifpsum\ extends the semantics of \fosum.
To define the semantics of
$
  \eta\coloneqq\ifp\big(
  F(\vec x)\gets\theta\big)(\vec x')
  $ on $(\A,\nu)$, with $\A$ an $\Upsilon$-structure with
$\ext(\eta)\subseteq\Upsilon$, 
we define a sequence $(F^{(t)})_{t\in\Nat\cup\{\infty\}}$ of
functions $F^{(t)}:A^{\ar(F)}\to\Rbot$ as follows.
We set
$F^{(0)}(\vec a)\coloneqq\bot$ for all $\vec a\in A^{\ar(F)}$, and 
$ F^{(t+1)}(\vec a)\coloneqq
  \sem{\theta}{(\A\frac{F^{(t)}}{F},\ass\frac{\vec a}{\vec x})}$
 if $F^{(t)}(\vec a)=\bot$, and $F^{(t)}(\vec a)$ otherwise.
Here, $\A\frac{F^{(t)}}{F}$ denotes the $\Upsilon\cup\{F\}$-structure that
coincides with $\A$ on all symbols in $\Upsilon\setminus\{F\}$ and
interprets $F$ by $F^{(t)}$. For $t=\infty$ we observe that for every $\vec a\in
A^{\ar(F)}$, if there is some $t$
such that $F^{(t)}(\vec a)\neq\bot$ then $F^{(t')}(\vec a)=F^{(t)}(\vec a)$ for all
$t'\ge t$, and in this case we let
$F^{(\infty)}(\vec a)\coloneqq F^{(t)}(\vec a)$. Otherwise, we let $F^{(\infty)}(\vec a)\coloneqq\bot$.
Finally, we let
$
  \sem{\eta}{(\A,\ass)}\coloneqq F^{(\infty)}\big(\ass(\vec x')\big).
$

We use the name \ifpsum\ both for the logic in the previous
section as well as the variant defined here. As we will show
next, the logics are indeed essentially the same. If we
explicitly need to distinguish between them,
we speak of \emph{\ifpsum\ strata} and
\emph{programs} for the syntax defined in Section~\ref{sec:ifpsum}
and of \emph{\ifpsum\ (weight) terms} and \emph{formulas}
for the syntax defined here.

We begin by observing that the
semantics of the ifp operator is compatible with the semantics of
strata from Section~\ref{sec:ifpsum}.  Specifically,
let $\theta(\vec x)$ be an \fosum-term and consider the
\ifpsum\ stratum $\Sigma\coloneqq F(\vec x)\gets\theta$,
where $\ext(\theta)\setminus\{F\}\subseteq\Upsilon$.
Then for all
$\Upsilon$-structures $\A$ we have
\[
  \sem{\ifp\big(
  F(\vec x)\gets\theta\big)}{\A}=F^{\Sigma(\A)}.
\]
Here we view $\sem{\ifp\big(
  F(\vec x)\gets\theta\big)}{\A}$ as the function from $A^{\ar(F)}$ to $\Rbot$ mapping a tuple $\vec a'$ to 
  $\sem{\ifp\big(F(\vec x)\gets\theta\big)(\vec x')}{\A}(\vec a')=\sem{\ifp\big(
  F(\vec x)\gets\theta\big)(\vec x')}{(\A,\ass\frac{\vec a'}{\vec x'})}$ for all assignments $\ass$.

An apparent difference, however, between programs and ifp-terms, is
that each ifp-term defines a single intensional weight function, while
programs can inductively define multiple intensional relations and
weight functions.  This does not imply a higher expressivity for
programs, however.  By adapting the proof of the Simultaneous
Induction Lemma \cite{moschovakis,ef_fmt2} to the \ifpsum\ setting, we
can show the following.

\begin{lemma} \label{simindlemma}
  For every \ifpsum\ program $\Pi$ with answer symbol $S$ there is a
  closed\/ \ifpsum\ expression $\xi$
  such that $ \sem{\xi}{\A}=S^{\Pi(\A)}$
  for all\/ $\Upsilon$-structures $\A$.
\end{lemma}

The remaining difference between ifp-terms 
 and programs is that ifp-operators can be nested, but strata
 can only be composed sequentially.  We can, however, show the
 following normal form.
Let us say that a \emph{selection condition} is a conjunction of
equalities and inequalities.

 \begin{lemma}\label{lem:normal-form}
   Every \ifpsum\ formula $\varphi(\vec x)$, as well as
     every \ifpsum\ program with a relation symbol as answer
       symbol, is equivalent to a formula of the form
       $
       \exists \vec y\big(\chi(\vec y)\wedge\ifp\big(F(\vec x,\vec y)\gets\theta\big)(\vec
       x,\vec y)\big) $,
     where $\chi$ is a selection condition and $\theta(\vec x,\vec y)$ is
     an \fosum\ term.
   Also, every \ifpsum\ term $\eta(\vec x)$, as well as every
     \ifpsum\ program with a weight function symbol as answer
       symbol, is similarly equivalent to a term of the form
       $\avg_{\vec y:\chi(\vec y)}\ifp\big(F(\vec x,\vec y)\gets\theta\big)(\vec
       x,\vec y)$.
 \end{lemma}

   This lemma can be proved along the lines of the analogous
   normal-form result for inflationary fixed-point logic
\cite[Chapter~8]{ef_fmt2}. 
   The idea of the proof is to first simulate nested $\ifp$-operations
   by a simultaneous induction, as it is defined by an \ifpsum\
   stratum. In the simultaneous induction, we first step through the
   inner induction until it reaches a fixed point. Then we take a
   single step of the outer induction, again run through the inner
   induction till it reaches a fixed-point, take a step of the outer
   induction, et cetera, until the outer induction reaches a fixed
   point. For this to work, it is crucial that we can detect if an
   induction has reached a fixed point in \fosum\ and then set a flag
   to indicate that to the outer induction.

 \begin{corollary}
   \begin{enumerate}
   \item Every \ifpsum\ program whose answer symbol is a relation
     symbol is equivalent
     to an \ifpsum\ formula of the form
     \[
       \exists \vec y\Big(\chi(\vec y)\wedge\ifp\big(F(\vec x,\vec y)\gets\theta\big)(\vec
       x,\vec y)\Big),
     \]
     where $\chi$ is a selection condition and $\theta$ is
     an \fosum\ term.
   \item Every \ifpsum\ program whose answer symbol is a function
     symbol is equivalent
     to a term of the form
     \[
       \avg_{\vec y:\chi(\vec y)}\ifp\big(F(\vec x,\vec y)\gets\theta\big)(\vec
       x,\vec y),
     \]
     where $\chi$ is a selection condition and $\theta$ is
     an \fosum\ term.
   \end{enumerate}
 \end{corollary}

  \begin{remark}
    Since our motivation is neural networks, we fixed
    the numerical domain to the reals.
However, the results presented in this section
generalise to any
    numerical domain with aggregates
    \cite{gg_metafinite,libkin_sql}.  The only requirement is that
    the logic, without recursion, can express the construct
    $\uniq_{x:\varphi} \theta$, with $\varphi$ a formula
    and $\theta$ a weight term, having the following semantics.
    Suppose there exists $a \in A$ such that $\A \models
    \varphi(a)$, and moreover, for all such $a$, the value
    $\sem\theta\A(a)$ is the same, say $r \in \Rbot$.
    Then we define
    $\sem{\uniq_{x:\varphi}\theta}{\A}$ to be this $r$.
    Otherwise, we define it to be $\bot$.

    In \fosum, we can indeed express this as
    \textsf{if $\forall x\forall
      x'((\varphi(x)\land\varphi(x'))\to\theta(x)=\theta(x'))$
      then $\avg_{x:\varphi}\theta$ else $\bot$},
    where by $\varphi(x')$ and $\theta(x')$ we mean that $x'$ (a
    fresh variable) is substituted for the free occurrences of $x$.
  \end{remark}

\section{A Polynomial Time Fragment} \label{secscalar}

In finite model theory, fixed-point logics are typically used to
capture polynomial-time computations, and it would be nice if we could
use \ifpsum\ to capture the polynomial properties of neural
networks.  We have to restrict our attention to weighted
structures with rational weights if we want to do this, at least if we
want to work in a traditional computation model. But even then
it turns out that \ifpsum\ terms cannot be
evaluated in polynomial time, as their values may get too large to
even be represented in space polynomial in the size of the
input structure.

\begin{example}[\cite{gg_metafinite}] \label{exa:squaring}
  Consider the following \ifpsum\ term $\eta(x)$:
  \[
    \eta(x)\coloneqq
    \ifp\big(F(x)\gets
      \ite{\exists y\, E(y,x)}{\,(\sum_{y:E(y,x)} F(y))\cdot
       (\sum_{y:E(y,x)} F(y))}{2}\big)(x).
 \]
 Then if $\A$ is a path of length $n$ and $a$ is last vertex of this
 path, then $ \sem{\eta}{\A}(a)=2^{2^n} $.
\end{example}

To avoid repeated squaring as illustrated above,
we will define a fragment $\sifpsum$, called \emph{scalar
  $\ifpsum$}, that limits the way multiplication
and division can be used.
It forbids multiplication between two terms that both
contain intensional function symbols.  
In other words, we only allow scalar multiplication when
intensional weight functions are involved, where ``scalar'' is
interpreted liberally as ``defined nonrecursively''.
We also forbid division by recursively defined terms.

Formally, for a set $\CF$ of function symbols, the syntax of
\emph{$\CF$-scalar formulas and terms} follows exactly the
grammar \eqref{eq:fosum-formulas}, \eqref{eq:fosum-terms},
\eqref{eq:ifpsum-terms} of \ifpsum\ formulas and terms, except
that the rules for multiplication, division, and $\ifp$ are
changed as follows:
\begin{itemize}
\item if $\theta_1$ is $\CF$-scalar and $\theta_2$ is a term with
  $\ext(\theta_2)\cap\CF=\emptyset$, then $\theta_1\cdot\theta_2$,
  $\theta_2\cdot\theta_1$, and $\theta_1/\theta_2$ are $\CF$-scalar;
\item if $\theta$ is $\CF\setminus\{F\}$-scalar, then $\ifp\big(F(\vec x)\gets\theta\big)(\vec
  x')$ is $\CF$-scalar.
\end{itemize}
Now an \ifpsum\ expression $\xi$ is called \emph{scalar} if all its 
subexpressions are
$\ints(\xi)$-scalar.  We denote the scalar fragment of
\ifpsum\ by $\sifpsum$.
  Similarly, an \ifpsum\ stratum $\Sigma$ is scalar if
  all subexpressions in rules are $\ints(\Sigma)$-scalar,
and a program is scalar if all its strata are.

\begin{example}
  The term from
Example~\ref{exa:squaring} is not scalar as it contains
  multiplication of two terms involving the intensional symbol
  $F$.  The programs in
    Examples \ref{ex:receval} and \ref{floydwarshall}
    are scalar.  \qed
  \end{example}

When we think about the complexity of evaluating \ifpsum\ expressions,
we restrict our attention to structures with rational weights, which
for simplicity we call \emph{rational structures}. For a rational
$\Upsilon$-structure $\A$, by $\|\A\|$ we denote the bitsize of an
encoding of $\A$. Then
\[
  \|\A\|=O\big(|A|+\sum_{R\in\Upsilon\text{ relation
        symbol}}|R^\A|+\sum_{\substack{F\in\Upsilon\text{
          weight-}\\\text{function symbol}}}\sum_{\vec a\in
        A^{\ar(F)}}\|F^\A(\vec a)\|\big),
\]
where for an integer $n\in\Nat$, we let $\|n\|$ be the length of the
binary encoding of $n$ and for a rational $p/q\in\Q$ in reduced
form we let $\|p/q\|=\|p\|+\|q\|$.

We can show that scalar \ifpsum\ has polynomial data complexity.

\begin{theorem}\label{theo:eval-ifp}
  There is an algorithm that, given an $\sifpsum$ expression $\xi$,
  a rational structure $\A$, and a tuple $\vec a\in A^{|\vec
      x|}$, computes $\sem{\xi}{\A}(\vec a)$ in time
  polynomial in $\|\A\|$. 
\end{theorem}

To prove this theorem (see Appendix), we can start from the normal form of
Lemma~\ref{lem:normal-form}, which can be seen to preserve
scalarness.  The proof amounts to a careful verification that, in a recursive
rule $F(\xx) \gets \theta$, the numerator and denominator of
$\sem\theta\A(\bvec a)$
only depend linearly on $F^\A$, in a sense that can be made
precise.  The difficulty lies in controlling the growth of
denominators under the arithmetic operators.

\begin{remark}
In the language of parameterized complexity, the
\ifpsum-evaluation problem is in the complexity class XP if we
parameterized by the length of the expression. As \ifpsum\
contains the standard fixed-point logic UFP, the problem is
actually hard for (uniform) XP \cite{FlumG06}. 
\qed
\end{remark}

Even on unweighted relational structures, \sifpsum\ is more
expressive than plain fixed point logic \IFP,
because it can count the number of elements in
definable sets using the summation operator, then compare such
counts to express, e.g., the majority query.
Yet, unsurprisingly, we have the following even for full \ifpsum.

\begin{theorem}\label{theo:inexp}
  There is a boolean query on graphs that is decidable in polynomial
  time, but not expressible in \ifpsum.
\end{theorem}

This result can be proved by standard techniques; we give a
self-contained proof in the Appendix.
We show that the bijective
pebble game \cite{Hella96}, which is the standard tool for
proving inexpressibility in fixed-point logic with counting, can
also be used to prove inexpressibility in \ifpsum.  This was
already known for logics with aggregates quite similar to \fosum\
\cite{hlnw_aggregate}.  Then the usual Cai-Fürer-Immerman
construction \cite{cfi} provides an example of a query
inexpressible in  \ifpsum.

\begin{figure}
  \centering
  \begin{tikzpicture}[
        vertex/.style = {draw,circle,minimum size = 5pt,inner
          sep=0pt},
        internal/.style = {draw,fill,circle,minimum size = 4pt,inner
          sep=0pt},
        ]
    \begin{scope}
      \node[vertex] (in) at (0,0) {};
      \node[vertex] (out) at (0,1.6) {};
      \draw[thick,->] (in) edge node[right] {$a$} (out);
      \path (0,-0.5) node {(a)};
    \end{scope}

    \begin{scope}[xshift=4cm]
      \node[vertex] (in) at (0,0) {};
      \node[vertex] (out) at (0,1.6) {};
      \node[internal] (i1) at (-1,0.8) {}; 
      \node[internal] (i2) at (-0.5,0.8) {}; 
      \node[internal] (i3) at (1,0.8) {};
      \node at (0.25,0.8) {$\ldots$};
      \draw[->] (in) edge node[left] {$1$} (i1) edge
      node[right] {$1$} (i2) edge node[right] {$1$} (i3)
      (i1)  edge node[left] {$1$} (out) (i2)  edge node[right] {$1$}
      (out) (i3)  edge node[right] {$1$} (out);
      \path (0,-0.5) node {(b)};
    \end{scope}
    \begin{scope}[xshift=8cm]
      \node[vertex] (in) at (0,0) {};
      \node[vertex] (out) at (0,1.6) {};
      \node[internal] (i1) at (0,0.8) {}; 
      \draw[->] (in) edge node[right] {$1$} (i1) 
      (i1)  edge node[right] {$a$} (out);
      \path (0,-0.5) node {(c)};
    \end{scope}
  \end{tikzpicture}
  \caption{Splitting (a) an edge with large weight $a\in\Nat$ into (b) $a$
    internal nodes connected by edges of weight $1$. Reducing the net
    in (b) yields the net (c)}
  \label{fig:edge-split}
\end{figure}
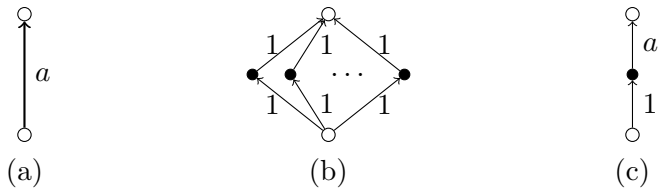

We are mainly interested, however, in model-agnostic queries on neural
networks. Nevertheless, we still cannot express all such
polynomial-time queries in \ifpsum.
An intuition for this is that weights in an FNN can be
arbitrary large and \ifpsum\ is too weak for
general (polynomial-time) computations with numbers.  
Actually, already the very simple query $Q_0$ where $Q_0(\Net)$ is
true iff $f^\Net(1)$ is a natural number,
can be shown to be not expressible.
But it gets worse: large weights can be avoided,
because we can split large-weight edges into many small-weight edges, as
Figure~\ref{fig:edge-split} illustrates.  Using the variation of
the above query $Q_0$ that asks if $f^\Net(1)$
is an even natural number, we can show the following.

\begin{theorem} \label{theor:not-ptime}
  There is a model-agnostic Boolean query on the class $\K(1,1)$ that is
  decidable in polynomial time, but not expressible in \ifpsum\ even
  on FNNs where all weights are $1$ or $0$.
\end{theorem}

The dependence on large weights indeed requires a more subtle
analysis.  Thereto, for every FNN $\Net$
we shall define an equivalent \emph{reduced} FNN $\RNet$.
Recall that $b(u)$ and $w(v,u)$ denote the bias of nodes $u$
and weight of edges $(v,u)$ in a network.
For a set $V$ of nodes, we let
$\wt(V,u)\coloneqq \sum_{v\in V}\wt(v,u)$.

We define an equivalence relation $\sim$ on the set of nodes of
$\Net$ by induction on the depth of the nodes in $\Net$. Only
nodes of the same depth can be equivalent. For input nodes and
output nodes $u$, $u'$ we let $u\sim u'\Leftrightarrow u=u'$, so 
input and output nodes are only equivalent to themselves. Now
consider two hidden nodes $u$, $u'$ of the same depth, and suppose
that we have already defined $\sim$ on all nodes of smaller
depth. Then
$u\sim u'$ if $\bias(u)=\bias(u')$ and $\wt(V,u)=\wt(V,u')$
for every equivalence class $V$ of nodes of smaller depth.

By $\tilde u$ we denote the equivalence class of node $u$.
We now define:
\begin{definition}
The reduced network $\RNet$ has nodes $\tilde u$ and edges $(\tilde
v,\tilde u)$ for all edges $(v,u)$ of $\Net$.
We define $\bias(\tilde u)\coloneqq \bias(u)$ and
$\wt(\tilde v,\tilde u)\coloneqq \sum_{v'\sim v}\wt(v',u)$.
The input nodes of $\RNet$ are the singleton classes of the input
nodes of $\Net$, and similarly for the output nodes.
\end{definition}

A straightforward induction shows that if $u\sim u'$ then
$f^\Net_u=f^\Net_{u'}$, so $f^\Net=f^{\RNet}$, i.e.,
an FNN and its reduction represent the same function.

\begin{example}
  The reduction of the network in Figure~\ref{fig:edge-split}(b)
  is shown in Figure~\ref{fig:edge-split}(c). \qed
\end{example}

Let $\pol(X)$ be a polynomial, and let $\Net$ be an FNN;
we will use the notation $|\Net|$ for the number of nodes of a network.
With $n=|\Net|$,
we say that $\Net$ has \emph{$\pol$-bounded
weights} if all node and edge weights are rational and of the
form $r/q$ for integers $r\in\mathbb Z$, $q\in\Nat$ such that
$|r|,q\le \pol(n)$.  Furthermore, we say that $\Net$ has
\emph{$\pol$-bounded reduced weights} if $\RNet$ has
$\pol$-bounded weights.  A class $\K$ of networks has
\emph{polynomially bounded (reduced) weights} if there exist a
polynomial $\pol$ so that every $\Net \in \K$ has $\pol$-bounded
(reduced) weights.

\begin{example}
  The class of networks depicted in
  Figure~\ref{fig:edge-split}(b) obviously has polynomially bounded
  weights (they are all $1$),
  but not polynomially bounded reduced weights. Indeed, the class
  depicted in 
  Figure~\ref{fig:edge-split}(c) does not have polynomially
  bounded weights (since $a$ can be arbitrary but $|\Net|=3$).
  \qed
\end{example}

We establish the following completeness result for scalar
\ifpsum.
\begin{theorem} \label{theo:scalarptime}
Let $Q$ be a polynomial-time computable query on $\K(*,*)$.
Then $Q$ is expressible in $\sifpsum$ on every
  $\K' \subseteq \K(*,*)$ with polynomially bounded reduced
  weights.
\end{theorem}

The proof (see Appendix) observes that a quasi-order on the nodes
of a network $\Net$ can be defined uniformly in $\sifpsum$ such
that it yields a linear order on $\RNet$.  Moreover, since
$\RNet$ has polynomially bounded weights, all numerators and
denominators can be represented, using $\sifpsum$ formulas,
by lexicographically ordered
tuples of nodes. This yields a copy of $\RNet$ as an
ordered finite structure defined in $\Net$.  By the
Immerman-Vardi theorem, we can then also define, in $\Net$, the
answer of $Q$ on $\RNet$.  As $Q$ is model-agnostic, this is also
the answer of $Q$ on $\Net$.

\begin{remark}
  For simplicity, Theorem~\ref{theo:scalarptime} is stated and
  proved for Boolean queries without parameters.  A version for $r$-ary
  queries with parameters can be formulated and proved, where we
  then also need to restrict to polynomially bounded reduced rational
  numbers for the parameters and result tuples.  The candidates for
  parameters and result tuples are passed to the $\sifpsum$
  formula as extra weight functions.
\end{remark}

\section{Complexity of model-agnostic queries} \label{secnph}

With \ifpsum\ in place, it is now tempting to revisit
Theorem~\ref{theorfixedepth} and wonder if the fixed-depth
restriction can be lifted simply by replacing \fosum\ by \ifpsum.
We conjecture, however, that the answer is negative:
\begin{conjecture} \label{conjecture}
  Let $m>0$ be a natural number.  There exists a Boolean 
  query on $\K(m,1)$ expressible in $\bblin$ but not in \ifpsum.
\end{conjecture}
While we cannot prove this conjecture,
we can prove the corresponding conjecture for
scalar \ifpsum, assuming $\PTIME \neq \NP$.
Indeed, whereas $\sifpsum$ queries are
computable in polynomial time, there
are very simple $\bblin$ sentences that already express an \NP-hard
query, even on $\K(1,1)$.

\begin{theorem}\label{thm:np-hardness-fx-zero}
It is \NP-hard to decide if an FNN $\Net \in\K(1,1)$ computes a
non-zero function, that is, if $f^\Net(x)\neq 0$ for some $x\in\R$.
\end{theorem}

Note that testing non-zeroness is a boolean model-agnostic query,
easily expressed by the $\bblin$ sentence $\exists x\, f(x)\neq
0$.  In fact, the query also belongs to \NP
\cite[Proposition~1]{Wurm24}.  In the cited work, Wurm also
proves \coNP-hardness of deciding whether an FNN $\Net \in
\K(*,*)$ is zero (phrased as an equivalence problem).  Our
contribution here is that hardness already holds when the input
and output dimensions are fixed to $1$.  The proof (see Appendix)
reduces from 3-SAT\@. We interpret rational numbers between 0 and
1, with with binary representation $(0.a_1a_2\ldots a_n)_2$, as
assignments on $n$ boolean variables.  We then construct a
network $\Net$ simulating, on these numbers, a given 3-CNF
formula over these variables.

\section{Iterated transductions} \label{seciter}

In view of Conjecture~\ref{conjecture}, how can we go beyond
\ifpsum\ to define a logic over weighted structures that can
express all $\bbl$ queries without fixing the network depth?  We
can get inspiration from how Theorem~\ref{theorfixedepth} was
proved \cite{ql4nn}.  The first step of that proof is to map any given
(fixed-depth) FNN $\Net \in \K(m,1)$ to a structure that
represents the geometry of the piecewise linear function
$f^\Net:\R^m\to\R$.  This geometry consists of a hyperplane
arrangement that partitions $\R^m$ in polytopes, plus an affine function on
each of the polytopes.  Later steps embed the resulting geometry
in a higher-dimensional space as dictated by the number of variables of
the $\bblin$ query $\psi$ that we want to express, and construct a
cylindrical cell decomposition of the space.  The resulting cell
decomposition is compatible both with $f^\Net$ and with the
constraints imposed by $\psi$, which allows us to express $\psi$
in \fosum.

The steps just described map weighted structures to weighted
structures; these transductions are expressed in \fosum\ by
adapting the classical model-theoretic method of interpreting one
logical theory in another \cite{hodges}.  Such interpretations
(also called transductions \cite{courcelle_book,grohe_book} or
translations \cite{ql4nn}) map
structures from one vocabulary to structures from another
vocabulary by defining the elements of the output structure as
equivalence classes of tuples from the input structure, and
defining the relations and weight functions by formulas and
weight terms.

The important observation is that, in showing that the steps of
the proof described above can be expressed as \fosum\
transductions, the assumption of a fixed depth $\ell$ of $\Net$
is only important for the first step.  That construction is
performed layer by layer, with a transduction that is iterated a
fixed number $\ell$ times.  We can thus lift the fixed-depth
restriction if we can iterate a transduction an unbounded number
of times; actually, $O(\ell)$ iterations suffice.

Formally, we can define an \emph{iterated \fosum\ transduction}
from vocabulary $\Upsilon$ to vocabulary $\Gamma$ as a pair
$(\tau,\varphi)$ where $\tau$ is an \fosum\ transduction from
$\Upsilon \cup \Gamma$ to itself, and $\varphi$ is a closed \fosum\ 
formula over $\Upsilon \cup \Gamma$.  The semantics, given an
$\Upsilon$-structure $\A$, is to first expand $\A$ with the symbols 
of $\Gamma$ (initialising them to be empty or undefined
everywhere).  Then, $\tau$ is repeatedly applied until $\varphi$
becomes true.

The difference with \ifpsum\ programs, even under loose semantics,
is that a transduction can grow the universe, and can arbitrarily
change (also shrink) relations and functions
This can then happen in
each step of an iterated transduction.

We can show that iterated transductions are
closed under sequential composition. To express a query, we can
designate an answer symbol, just like we did for \ifpsum\
programs.  We conclude:

\begin{proposition}
  Let $m$ be a natural number. Every boolean $\bblin$ query on
  $\K(m,1)$ is expressible by an iterated \fosum\ transduction
  that, on any $\Net \in \K(m,1)$, iterates only $O(\ell)$ times,
  where $\ell$ is the depth of $\Net$.
\end{proposition}

\begin{remark}
Iterated transductions are something of a ``nuclear option,'' in
that they are reminiscent of the extension
  of first-order logic with while-loops and object creation,
  investigated in the 1990s \cite{av_datalog,ak_iql,vdbp_abstr,vvag_compl}.  
  Such an extension typically yields a computationally complete query
  languages over finite unweighted structures.  We can show
  (proof omitted) that, similarly, iterated \fosum\ transductions (without a
  depth bound on the number of iterations, as in the above
  proposition) are computationally complete over rational weighted
  structures.
\end{remark}

\section{Conclusion} \label{seconc}

We have explored approaches, and obstacles, to model-agnostic
querying of deep neural networks.  In the fixed-depth
case, \fosum\ is already quite expressive, and the main challenges
now lie in finding good implementation strategies.  Without 
a bound on the depth, there are also challenges in expressivity,
theoretical complexity, and query language design.
We have introduced a language \ifpsum\ that enables us to
evaluate neural networks of unbounded depth and, more generally,
express a rich set of queries on such networks.

Some interesting questions remain open.
A very concrete question, even independent of the application to
neural networks, is to characterise the
data complexity of \ifpsum\ (without the scalar restriction),
even on unweighted structures.  We also encountered the question
whether the loose fixpoint semantics can keep arities lower than
possible with the function fixpoint semantics.
Another interesting question is
how the yardstick logic $\bbl$ can be adapted to work over
arbitrary functions $f:\R^m\to\R^p$ where $m$ and $p$ are not fixed in
advance.  Over Boolean models, there are very elegant
languages for this \cite{arenas-foil,arenas-dtfoil}.

Finally, and of course, we should also investigate the querying
of other machine-learning (ML) models, such as numerical decision
trees, Transformer models, and graph neural networks.  A large
body of work has accumulated in the ML, logic, and
database theory communities on understanding the logical
expressiveness of ML models.  Nevertheless, it is quite a
distinct subject to understand the querying of these models by
logical methods \cite{silva-logic-explain}.

\printbibliography

\newpage
\appendix
\section{Proofs}

In our proofs we use the notation $[n]$, with $n$ a natural
number, for $\{1,\dots,n\}$.

\subsection{Proof of Theorem~\ref{theor-loose}}

We shall prove that there exists a translation from the loose to the functional semantics in three parts. First we prove that such a translation exists if we assume that the structure always has at least two elements. Then we prove that a translation exists if we assume the structures have less than two elements, and then finally we show how to combine them into a translation that works for all structures.
\begin{lemma}
    \label{lem:loose_to_func_big_dom}
    For every stratum $\Sigma$ of type $\Upsilon \to \Gamma$
    there exists a stratum $\Sigma'$ of type $\Upsilon \to
    \Gamma'$, with $\Gamma \subset \Gamma'$, such that for any
    $\Upsilon$-weighted structure $\A$ with domain $A$ and $|A|
    \geq 2$, we have that ${\Sigma'}(\A)$, restricted to the
    symbols of\/ $\Upsilon \cup \Gamma$, equals $\Sigma^L(\A)$.
\end{lemma}
\begin{proof}
    The approach of this proof will be to simulate the evaluation of a single stage under loose semantics with two stages under functional semantics. We will first update all the intensional relations and then update the intensional weight functions, using the old values of the intensional relations. 
    Additionally, since the intensional weight functions can be entirely rewritten at each stage under loose  semantics, we parameterise the weight function with a timestamp, as previously explained. For our timestamp, we will use set of tuples that are added to the intensional relations in the corresponding stage under loose semantics. It is for the construction of this timestamp that we need to simulate once stage of the loose inflationary semantics with two stages under the functional semantics.

    Let $R_1, \dots, R_n$ be all the relation names in $\Gamma$.
    To start, we will construct a formula $\phiAll$ that accepts
    all concatenations of exactly one tuple in each of $R_1,
    \dots, R_n$. To allow for some relations to be empty, we will
    add ``empty tuples'' by encoding our tuples prefixed by a
    pair of extra elements. If the elements in this pair have the
    same value, it encodes ``no tuple'',
and if the elements in the prefix pair have the different values, it encodes an actual tuple. We can write $\phiAll$ as follows
    \begin{tabbing}
        \qquad\=\qquad\=\kill
        \(\phiAll(x_1,\dots, x_{s_{n+1}}) := \displaystyle \bigwedge_{i = 1}^{n}(x_{1+s_i} = x_{2+s_i} \lor (x_{1+s_i} \neq x_{2+s_i} \land R_i(x_{3+s_i},\dots, x_{s_{i+1}})))\)
    \end{tabbing}
    where $s_i :=\sum_{j=1}^{i-1} (\ar(R_{i})+2)$ for $i\in \{1,\dots,n+1\}$. The tuples accepted by this formula will form our timestamps.

    Next we will add the following relational rules to $\Sigma'$ to track the history of $\phiAll$ and all $R_1, \dots, R_n$:
    \begin{tabbing}
        \qquad\=\qquad\=\kill
        $\RAll(x_1,\dots, x_{s_{n+1}}) \gets \phiAll(x_1,\dots, x_{s_{n+1}})$ \\[\jot]
        $R_{i}^\text{old}(x_1,\dots,x_{\ar(R_{i})}) \gets R_i(x_1,\dots,x_{\ar(R_{i})})$ for each $i \in \{1,\dots,n\}$ \\[\jot]
        $\RAllOld(x_1,\dots, x_{s_{n+1}}) \gets \RAll(x_1,\dots, x_{s_{n+1}})$ \\[\jot]
        $\RAllOldOld(x_1,\dots, x_{s_{n+1}}) \gets \RAllOld(x_1,\dots, x_{s_{n+1}})$
    \end{tabbing}

    We now define the formulas $\phiAllNew$, $\phiAllOldNew$, and $\phiAllOldOldNew$ that contain the tuples that will be added to $\RAll$ this stages and those that were added to it one and two stages ago respectively. They are defined as follows
    \begin{tabbing}
        \qquad\=\qquad\=\kill
        \( \phiAllNew(x_1,\dots, x_{s_{n+1}}) := \phiAll(x_1,\dots, x_{s_{n+1}}) \land \neg \RAll(x_1,\dots, x_{s_{n+1}})\)\\[\jot]
        \( \phiAllOldNew(x_1,\dots, x_{s_{n+1}}) := \RAll(x_1,\dots, x_{s_{n+1}}) \land \neg \RAllOld(x_1,\dots, x_{s_{n+1}})\)\\[\jot]
        \( \phiAllOldOldNew(x_1,\dots, x_{s_{n+1}}) := \RAllOld(x_1,\dots, x_{s_{n+1}}) \land \neg \RAllOldOld(x_1,\dots, x_{s_{n+1}}) \)
    \end{tabbing}

    Let $w_1, \dots, w_m$ be all the weight function names in $\Gamma$. We will define a new intensional weight function name $w_{i}'$ of arity $s_{n+1}+\ar(w_i)$ for each $i \in \{1,\dots,m\}$, that will represent the weight function $w_i$ in the stage where the first $s_{n+1}$ variables were added to $\RAll$. Let $w_i(y_1,\dots,y_{\ar(w_i)}) \gets t_i$ be the weight function rule in $\Sigma$ for each $i\in\{1,\dots,m\}$. We define $t_{i}'$ to be the weight term $t_i$ where each occurrence of $R_j$ is replaced by $R_{j}^{\text{old}}$ for each $j \in \{1,\dots,n\}$ and where each occurrence of $w_k(z_1,\dots,z_{\ar(w_k)})$ is replaced by 
    \begin{equation*}
        \avg_{p_1,\dots, p_{s_{n+1}}: \phiAllOldOldNew(p_1,\dots, p_{s_{n+1}})} w_{k}'(p_1,\dots, p_{s_{n+1}},z_1,\dots,z_{\ar(w_k)})
    \end{equation*}
    for each $k\in\{1,\dots,m\}$. We now add the following weight function rule to $\Sigma'$
    \begin{equation*}
        w_{i}'(x_1,\dots, x_{s_{n+1}},y_{1},\dots,y_{\ar(w_i)}) \gets \IfText \phiAllNew(x_1,\dots, x_{s_{n+1}}) \ThenText t_{i}' \ElseText \bot
    \end{equation*}
    which makes sure that the $w_{i}'$ can only be updated during the stage in which we are updating the intensional weight functions. This is because $\phiAllNew$ only contains tuples if some tuple was added to any intensional relation in the previous stage, because a stage in which we update the intensional relations is preceded by a stage in which we update only the intensional weight functions.

    Next, we will add the relation rules for $R_1,\dots,R_n$. To do this, we need to make sure we do not update the intensional relations while the intensional weight functions are being updated. To this end we will define the formula $\phiWUpdate$ as follows:
    \begin{equation*}
        \phiWUpdate:= \exists x_1,\dots, x_{s_{n+1}}\,\phiAllNew(x_1,\dots, x_{s_{n+1}})
    \end{equation*}
    Let $R_i(x_1,\dots,x_{\ar(R_{i})}) \gets \varphi_i(x_1,\dots,x_{\ar(R_{i})})$ be the relational rule for $R_i$ in $\Sigma$ for each $i \in \{1,\dots,n\}$. We define $\varphi_{i}'(x_1,\dots,x_{\ar(R_{i})})$ to be the formula $\varphi_i$ where each occurrence of $w_j(y_1,\dots,y_{\ar(w_{j})})$ is replaced by 
    \begin{equation*}
        \avg_{p_1,\dots, p_{s_{n+1}}: \phiAllOldNew(p_1,\dots, p_{s_{n+1}})} w_{j}'(p_1,\dots, p_{s_{n+1}},y_1,\dots,y_{\ar(w_{j})}) 
    \end{equation*}
    for each $k\in\{1,\dots,m\}$. We now add the following relational rule to $\Sigma'$
    \begin{equation*}
        R_i(x_1,\dots,x_{\ar(R_{i})}) \gets \neg \phiWUpdate \land \varphi_{i}'(x_1,\dots,x_{\ar(R_{i})})
    \end{equation*}
    for each $i \in \{1,\dots,n\}$.

    Now we will finally add the weight function rules for $w_1,\dots,w_m$. Before we can do this we need to define the formula $\phiStop$ that is only true when $R_1,\dots,R_n$ have stopped updating. 
    The formulas $\phiAllNew$ and $\phiAllOldNew$ track the new tuples that will be added to $\RAll$, and those that were added to it 1 stage ago respectively. If no relation was updated during a relation update stage, in the next stage, a weight update stage, the formula $\phiAllNew$ will be empty. This alone is insufficient to check that the relations have stopped growing, because in every relation update stage the $\phiAllNew$ is empty, since the relations are not updated during the preceding weight update stage. Thus, we can define $\phiStop$ as
    \begin{equation*}
        \phiStop := \neg \exists x_1,\dots, x_{s_{n+1}} (\phiAllNew(x_1,\dots, x_{s_{n+1}}) \lor \phiAllOldNew(x_1,\dots, x_{s_{n+1}}))
    \end{equation*}
    Next, since $\phiStop$ will become true during a weight update stage, we will use $\phiAllOldOldNew$ to determine the identifier of the final stage. Thus, we will add the following rules to $\Sigma'$
    \begin{tabbing}
        \quad\=\quad\=\kill
        \( w_i(x_1,\dots,x_{\ar(w_i)}) \gets\)\\[\jot] 
        \> \( \displaystyle \IfText \phiStop \ThenText \avg_{p_1,\dots, p_{s_{n+1}}: \phiAllOldOldNew(p_1,\dots, p_{s_{n+1}})} w_{i}'(p_1,\dots, p_{s_{n+1}},x_1,\dots,x_{\ar(w_i)}) \ElseText \bot \)
    \end{tabbing} 
    for each $i \in \{1, \dots m\}$, where $b_i$ is the arity of $w_i$.
\end{proof}
To complete our proof translating from loose to functional semantics, we need to give a translation that works for weighted structures with a domain size smaller than 2. However, before that we will give a quick definition of a bounded loose termination index.
\begin{definition}
    Let $\Sigma$ be an \ifpsum{} stratum of type $\Upsilon \to \Gamma$ and let $\C$ be a class of $\Upsilon$-weighted structures.
    We say that the loose termination index is \emph{$\C$-bounded} if there exists some $n$ such that for the loose termination index of $\Sigma(\A)$ is less than or equal to $n$ for each $\A \in \C$.
\end{definition}
\begin{lemma}
    \label{lem:loose_to_func_small_dom}
    For any \ifpsum{} stratum $\Sigma$ of type $\Upsilon \to \Gamma$ and class of $\Upsilon$-weighted structures $\C$, if the loose termination index is $\C$-bounded, there exists a set of \fosum{} formulas and weight terms such that for each $\A \in \C$ the following holds:
    \begin{itemize}
        \item For each relation name $R \in \Gamma$ there exists a formula $\varphi_{R}$ of the same arity as $R$ such that $\varphi_{R}^{\A}$ agrees with $R^{\Sigma^L(\A)}$.
        \item For each weight function name $w \in \Gamma$ there exists a weight term $t_{w}$ with the same number of free variables as the arity of $w$ such that $t_{w}^{\A}$ agrees with $w^{\Sigma^L(\A)}$.
    \end{itemize}
\end{lemma}
\begin{proof}
    Let $n$ be a natural number smaller than or equal to one more than the bound on the loose termination index of $\Sigma$. First we will prove inductively that we can write a formula $\varphi_{R}^n$ for each relation name $R \in \Gamma$ and a weight term $t_{w}^n$ for each weight function name $w\in \Gamma$ such that $\varphi_{R}^n$ and $R^{\B_n}$ agree, and that $t_{w}^n$ and $w^{\B_n}$ agree for each $\Upsilon$-weighted structure $\A$, where $\B_n := L_{\Sigma}^n(\A)$. For the base case, where $n=0$, this is trivial since each term is $\bot$ and each formula is simply false. Let the $R \in \Gamma$ be a relation name and let $R(x_1,\dots,x_m) \gets \varphi(x_1,\dots,x_m)$ be its rule in $\Sigma$. Then for any $n>0$ we have that $\varphi_{R}^{n}$ is $\varphi$ with every occurrence of every relation $R' \in \Gamma$ replaced by $\varphi_{R'}^{n-1}$ and every weight function $w' \in \Gamma$ replaced by $t_{w'}^{n-1}$, which are all constructed inductively. We can similarly construct $t_{w}^{n}$ for each weight function name $w \in \Gamma$.
    
    Next we will prove that we can construct a formula $\lti{n}$, such that for each $\Upsilon$-weighted structure, $\A \models \lti{n}$ if and only if $n$ is the loose termination index of $\Sigma$ evaluated on $\A$. We will again prove this inductively. Let $\{R_1,\dots,R_m\}$ be the set of all relation names in $\Gamma$ and let $a_1,\dots,a_m$ be their respective arities. For the base case where $n=0$, we can define the formula simply as
    \begin{equation*}
        \lti{0} := \bigwedge_{R_i \in \{R_1,\dots,R_m\}} \neg \exists x_1,\dots ,x_{\ar(R_{i})}\,\varphi_{R}^1(x_1,\dots,x_{\ar(R_{i})})
    \end{equation*}
    For the case where $n>0$, we can define the formula as
    \begin{align*}
        \lti{n} := & (\bigwedge_{R_i \in \{R_1,\dots,R_m\}} \neg \exists x_1,\dots ,x_{\ar(R_{i})}(\varphi_{R}^n(x_1,\dots,x_{\ar(R_{i})}) \land \neg \varphi_{R}^{n-1}(x_1,\dots,x_{\ar(R_{i})})))\\
        & \land (\bigwedge_{i=0}^{n-1} \neg \lti{i})
    \end{align*}
    where all $\lti{i}$ is defined inductively for all $i < n$.

    Finally, we can combine the above two proofs to construct our desired formulae and weight terms as
    \begin{align*}
        \varphi_{R_i}(x_1,\dots,x_{\ar(R_{i})}) :=& \bigvee_{k=0}^{T} (\lti{k} \land \varphi_{R_i}^{k}(x_1,\dots,x_{\ar(R_{i})}))\\
        t_{w_j}(x_1,\dots,x_{\ar(w_{j})}) := &\IfText \lti{0} \ThenText t_{w_j}^{0}(x_1,\dots,x_{\ar(w_{j})})\\
        & \ElseTextNS \IfText \lti{1} \ThenText t_{w_j}^{1}(x_1,\dots,x_{\ar(w_{j})})\\
        & \vdots\\
        & \ElseTextNS \IfText \lti{T} \ThenText t_{w_j}^{T}(x_1,\dots,x_{\ar(w_{j})})\\
        & \ElseTextNS \bot
    \end{align*}
    for each relation name $R_i \in \Gamma$ with arity $\ar(R_{i})$ and for each weight function name $w_j \in \Gamma$, with $T$ an upper bound on the loose termination index of $\Sigma$.
\end{proof}
\begin{proposition}
    \label{prop:loose_to_func}
    For any \ifpsum{} stratum $\Sigma$ of type $\Upsilon \to \Gamma$, there exists an \ifpsum{} stratum $\Sigma'$ of type $\Upsilon \to \Gamma'$, with $\Gamma \subset \Gamma'$, such that for any $\Upsilon$-weighted structure $\A$, we have that ${\Sigma'}(\A)$, restricted to $\Upsilon \cup \Gamma$, equals $\Sigma^L(\A)$.
\end{proposition}
\begin{proof}
    Lemmas \ref{lem:loose_to_func_big_dom} and \ref{lem:loose_to_func_small_dom} can be combined into a single stratum. Let $R_i$ be a relation in $\Gamma$ and let $R_i(x_1,\dots,x_{\ar(R_{i})}) \gets \varphi_{R_i}^{\text{big}}(x_1,\dots,x_{\ar(R_{i})})$ be its relational rule in the stratum of Lemma~\ref{lem:loose_to_func_big_dom}. Similarly, let $w_j$ be a weight function in $\Gamma$ and let $w_j(x_1,\dots,x_{\ar(w_{j})}) \gets t_{w_j}^{\text{big}}(x_1,\dots,x_{\ar(w_{j})})$. We start by adding all the rules for symbols that are not in $\Gamma$ from the stratum from Lemma~\ref{lem:loose_to_func_big_dom} to $\Sigma'$. We then add the following rules to $\Sigma'$
    \begin{tabbing}
        \quad\=\quad\=\kill
        \(R_i(x_1,\dots,x_{\ar(R_{i})}) \gets\)\\[\jot]
        \> \(((\exists y_1,y_2\, y_1 \neq y_2) \land \varphi_{R_{i}}^\text{big}(x_1,\dots,x_{\ar(R_{i})})) \)\\[\jot]
        \> \({} \lor (\neg(\exists y_1,y_2\, y_1 \neq y_2) \land \varphi_{R_{i}}^\text{small}(x_1,\dots,x_{\ar(R_{i})}))\) \\[\jot]
        \(w_j(x_1,\dots,x_{\ar(w_{j})}) \gets\)\\[\jot] 
        \> \( \IfText \exists y_1,y_2\, y_1 \neq y_2 \ThenText t_{w_{i}}^{\text{big}}(x_1,\dots,x_{\ar(w_{j})}) \ElseText t_{w_{i}}^{\text{small}}(x_1,\dots,x_{\ar(w_{j})})\)
    \end{tabbing}
    for each such $R_i$ and $w_j$ in $\Gamma$, where $\varphi_{R_i}^{\text{small}}$ and $t_{w_j}^{\text{small}}$ are the \fosum{} formula and weight term for $R_i$ and $w_i$ from Lemma~\ref{lem:loose_to_func_small_dom} respectively. Since these formulas and weight terms are only considered whenever the size of the domain is less than or equal to 1, we know that that case the loose termination index is bounded by the number of relations and thus that these formulas and weight terms can be constructed.
\end{proof}

\subsection{Proof of Lemma~\ref{simindlemma}}
  It clearly suffices to prove this for a program with only a single
  stratum. Moreover, since in structures with just one element, all
  \ifpsum\ programs are easily seen to be equivalent to \fosum\
  expressions, without loss of generality we only consider structures
  $\A$ with at least $2$ elements.

  So let $\Sigma$ be a stratum of type $\Upsilon\to\Gamma$ consisting of the rules
  $R_i(\vec x_i)\gets\psi_i$ for $i\in[k]$ and $F_j(\vec
  x_j)\gets\eta_j$ for $j\in[\ell]$. Without loss of generality we assume that $k,\ell\ge 1$; otherwise we introduce dummy
  rules. For $i\in[k]$, let $r_i$ be the arity of $R_i$, and for
  $j\in[\ell]$, let $s_j$ be the arity of $F_j$. Moreover, let $r\coloneq\max\{r_1,\ldots,r_k,s_1,\ldots,s_\ell\}$. 

  For all $i\in[k+\ell]$, let
  \[
    \chi_i(z_1,\ldots,z_{k+\ell})\coloneqq
    z_i\neq
    z_{i'}\wedge\bigwedge_{j\in[k+\ell]\setminus\{i\}}z_j=z_{i'},
  \]
  where $i'=1$ if $i\neq 1$ and $i'=2$ if $i=1$. In the following,
  $\vec z$ always ranges over $(k+\ell)$-tuples
  $(z_1,\ldots,z_{k+\ell})$.
  We will use $(k+\ell)$-tuples $\vec c$ to represent indices in
  $[k+\ell]$, where $\vec c$ represents $i$ if it satisfies
  $\chi_i(\vec c)$,
  that is, if the $i$th entry is distinct from all others and if all
  entries except the $i$th are equal. Then with a single
  $(r+k+\ell)$-ary function $G$ we can represent $(k+\ell)$
  $r$-ary functions $G_i$, where
  \[
    G_i(\vec a)=b\iff G(\vec a,\vec c)=b\quad\text{for all $\vec c$
      satisfying $\chi_i(\vec c)$}.
  \]
  If some of the functions $G_i$ have smaller arity, we can also
  represent them, simply ignoring the arguments that are not needed.
  Furthermore, we can represent relations via their characteristic
  functions.

  Let $F\not\in\Upsilon$ be a
  fresh function symbol of arity $(k+\ell+r)$. We will use $F$ to
  represent the relations $R_i$ and the functions $F_i$ in the way just described.
  For every $i\in[k]$, we let $\psi_i'$ be the formula obtained from
  $\psi_i$ by replacing each
  subformula $R_p(x_1',\ldots,x_{r_p}')$ by the formula 
  \[
    \rho_p(x_1',\ldots,x_{r_p}')\coloneqq \exists
    x_{r_p+1}\ldots\exists x_{r}' \exists\vec z\big(\chi_p(\vec z)\wedge
    F(x_1',\ldots,x_r',\vec z)=1\big)
  \]
  and by replacing each subterm $F_q(x_1',\ldots,x_{s_q}')$ by the
  term
  \[
    \tau_q(x_1',\ldots,x_{s_q}')\coloneqq\avg_{(x_{s_q+1}',\ldots,x_r',\vec z):
      \chi_q(\vec z)} F(x_1',\ldots,x_r',\vec z).
  \]
 Similarly, for every $j\in[\ell]$, we let $\eta_j'$ be the term obtained from
  $\eta_j$ by replacing each
  subformula $R_p(x_1',\ldots,x_{r_p}')$ by the formula
  $\rho_p(x_1',\ldots,x_{r_p}')$ and each subterm $F_q(x_1',\ldots,x_{s_q}')$ by the
  term $\tau_q(x_1',\ldots,x_{s_q}')$.

  Now for every $i\in[k+\ell]$ we define a term $\theta_i(\vec x,\vec z)$ as follows:
  \begin{itemize}
  \item we let
    \[
      \theta_{k+\ell}(\vec x,\vec z)\coloneqq \textsf{if $\chi_{k+\ell}(\vec
        z)$ then $\eta_{k+\ell}'$ else $\bot$};
    \]
  \item for $j\in[\ell-1]$, we let
    \[
      \theta_{k+j}(\vec x,\vec z)\coloneqq \textsf{if $\chi_{k+j}(\vec
        z)$ then $\eta_{k+j}'$ else $\theta_{k+j+1}$};
    \]
  \item for $i\in[k]$, we let
    \[
      \theta_i(\vec x,\vec z)\coloneqq\textsf{if $\chi_i(\vec
        z)\wedge\varphi_i'$ then $1$ else $\theta_{i+1}$}.
    \]
  \end{itemize}
  Let $\theta\coloneqq\theta_1$ and consider the \ifpsum\ term
  \[
    \zeta(\vec x,\vec z)\coloneqq \ifp\big(F(\vec x,\vec z)\gets\theta\big)(\vec
    x,\vec z).
  \]
  Let $\A$ be a structure (with at least $2$ elements). As all free variables of the term $\theta$
  appear in $\vec x$ or $\vec z$, we do not need an assignment to
  define the functions $F^{(n)}:A^{k+\ell+r}\to\Rbot$ for
  $n\in\Nat$. By induction $n$, it is easy to prove that for all
  $i\in[k]$ and $\vec a^=(a_1,\ldots,a_r)\in A^{r}$ and $\vec c\in
  A^{k+\ell}$ such that $\A\models\chi_i(\vec c)$ we have
  \[
    F^{(n)}(\vec a,\vec c)=1\iff (a_1,\ldots,a_{r_i})\in
    R_i^{T^n_\Sigma(\A)}
  \]
  and for all 
  $j\in[\ell]$ and $\vec a=(a_1,\ldots,a_r)\in A^{r}$ and $\vec c\in
  A^{k+\ell}$ such that $\A\models\chi_{k+j}(\vec c)$ we have
  \[
    F^{(n)}(\vec a,\vec c)=F_j^{T^n_\Sigma(\A)}(a_1,\ldots,a_{s_j}).
  \]
  Furthermore, for all 
  $\vec a\in A^{r}$ and all $\vec c\in
  A^{k+\ell}$ such that $\A\not\models\chi_{i}(\vec c)$ for any
  $i\in[k+\ell]$ we have $F^{(n)}(\vec a,\vec c)=\bot$.

  To complete the proof, we need to make a case distinction depending
  on the answer symbol $S$ of $\Sigma$. If $S=R_i$ for some
  $i\in[k]$, we let
  \[
    \xi(x_1,\ldots,x_{r_i})\coloneqq\exists x_{r_i+1}\ldots\exists x_r\exists\vec
    z\Big(\chi_i(\vec z)\wedge\zeta(x_1,\ldots,x_r,\vec z)=1\Big).
  \]
  If $S=F_j$ for some $j\in[\ell]$, we let
  \[
    \xi(x_1,\ldots,x_{s_j})\coloneqq \avg_{(x_{s_j+1},\ldots,x_r,\vec z):
      \chi_j(\vec z)} \zeta(x_1,\ldots,x_r,\vec z).
  \]

\subsection{Proof of Theorem~\ref{theo:eval-ifp}}

As a starting point, we note that it is straightforward to show that 
$\fosum$ has polynomial-time data complexity. 

\begin{lemma}\label{lem:eval-fo}
    There is an algorithm that, given an $\fosum$ expression $\xi(\vec
    x)$, a rational structure $\A$, and a tuple $\vec a\in A^{|\vec
      x|}$, computes $\sem{\xi}{\A}(\vec a)$ in polynomial time in $\|\A\|$. 
\end{lemma}

\begin{proof}[Proof of Lemma]
All the arithmetic on rationals that needs to be carried out in polynomial time in terms of the bit-size of the input numbers. So we can evaluate terms in polynomial time. then we can evaluate first-order formulas in the usual way.
\end{proof}

Now to the proof of Theorem~\ref{theo:eval-ifp}.
  Inspection of the proof of Lemma~\ref{lem:normal-form} shows that
  the transformation of an expression into its normal form
  preserves scalarness.
  Hence it suffices to consider
  expressions in normal form, which basically means that we have to
  evaluate terms
  \begin{equation}
    \label{eq:2}
    \eta=\ifp\big(F(\vec x)\gets\theta\big)(\vec x'),
  \end{equation}
  where $\theta(\vec x)$ is an \fosum\ term with $\deg_{\{F\}}(\theta)\le
  1$. Let $\Upsilon\coloneqq\ext(\eta)=\ext(\theta)\setminus\{F\}$ and $k\coloneqq\ar(F)$.

 A \emph{common denominator} for a function $f:D\to\mathbb Q$ is a
 $q\in\PNat$ such that $f(d)\cdot q\in\Nat$ for all $d$. Note that if
 $D$ is finite there is a unique least common denominator for
 $f$. If $D=\emptyset$ we let $1$ be the least common
 denominator for $f$ by default. For $f:D\to\Qbot$ we define a \emph{(least)
   common denominator} to be a (least) common
 denominator for the restriction of $f$ to the set of all $d\in D$
 with $f(d)\neq\bot$.
 
  \medskip
  \textsc{Claim~1. }
  {\itshape 
  Let $\A$ be a rational $\Upsilon$-structure, and let 
  $\zeta(\vec y)$ be an \fosum\ term of vocabulary
  $\Upsilon\cup\{F\}$, where $|\vec y|=\ell$, such that $\deg_{\{F\}}(\zeta)\le 1$.
  
  Then there are $c,d\in\Nat$ such that
  $\|c\|,\|d\|\in\|\A\|^{O(1)}$ and the following holds for every function $\overline
  F:A^k\to\mathbb Q_\bot$. Let $q$ be a common denominator of $\overline
  F$, and let $p\in\PNat$ such that $p\ge |\overline F(\vec
  a)|\cdot q$ for all $\vec a\in A^k$ with $\overline F(\vec
  a)\neq\bot$.
  
  Then there is an $j\in[d]$ such that $j\cdot q$ is a common
  denominator for $\sem{\zeta}{(\A,\overline F)}$, viewed as a
  function from $A^\ell$ to $\Qbot$, and for all
  $\vec b\in A^\ell$ with
  $\sem{\zeta}{(\A,\overline F)}(\vec b) \neq\bot$ it holds that
  $\big|\sem{\zeta}{(\A,\overline F)}(\vec b)\big|\cdot j\cdot q\le
  c\cdot p$.
}

\medskip
\textsc{Proof of Claim~1:}
We prove the claim by induction on $\zeta$.
The base cases are straightforward:
\begin{itemize}
\item if $\zeta=r$ for a constant $r=\frac{p'}{q'}\in\Q$, we let $c\coloneqq p'$ and $d\coloneqq q'$, and if $\zeta=\bot$ we let $c\coloneqq0$ and $d\coloneqq 1$.;
\item if $\zeta=F(\vec y)$ we let $c\coloneqq d\coloneqq 1$;
\item if $\zeta=G(\vec y)$ for some
  weight-function symbol $G\in\Upsilon$ we let $d$ be the least
  common multiple of $G^\A$, and we let
  $c\coloneqq d\cdot\max\{|G^\A(\vec b)| : \vec b\in A^{\ar(G)}\}$.
\end{itemize}
Suppose next that $\zeta=\zeta_1\circ\zeta_2$, and let $c_i,d_i$ the constants for
$\zeta_i$ that we get from the induction hypothesis.
\begin{itemize}
\item If
  $\circ\in\{+,-\}$, we let $d\coloneqq d_1d_2$ and $c\coloneqq
  c_1d_2+c_2d_1$.
\item If $\circ=\cdot$, then at most one $\zeta_i$
  contains $F$, because $\deg_{\{F\}}(\zeta)\le 1$. Say, $\zeta_2$
  does not contain $F$. Then without loss of generality we may assume
  that $d_2$ is a common denominator for $\sem{\zeta_2}{\A}$ and $c_2\coloneqq d_2\cdot\max\big\{|\sem{\zeta_2}{\A}(\vec a)|\bigmid \vec a\in A^k\text{ with }\sem{\zeta_2}{\A}(\vec a)\neq\bot\big\}$.
  We let
  $c\coloneqq c_1c_2$ and $d\coloneqq d_1d_2$.
\item If $\circ=/$, then $\zeta_2$ does not contain $F$, and again we
  may assume that $d_2$ is a common denominator for
  $\sem{\zeta_2}{\A}$ and
  $c_2\coloneqq d_2\cdot\max\big\{|\sem{\zeta_2}{\A}(\vec a)| : \vec a\in A^k\text{ with }\sem{\zeta_2}{\A}(\vec a)\neq\bot\big\}$. We let
  $c\coloneqq c_1d_2$ and $d\coloneqq d_1c_2$.
\end{itemize}
  Suppose next that
  $\zeta=\ite{\varphi}{\zeta_1}{\zeta_2}$, and let $c_i,d_i$ the constants for
  $\zeta_i$ that we get from the induction hypothesis. We let
  $d\coloneqq d_1d_2$ and $c\coloneqq\max\{c_1d_2,c_2d_1\}$.

  Finally, suppose that $\zeta(\vec y)=\sum_{\vec
    z:\varphi(\vec y,\vec z)}\zeta'(\vec y,\vec z)$, where
    $|\vec z|=m$. and let $c',d'$ the constants for
  $\zeta'$ that we get from the induction hypothesis. Let $j\le
  d'$ such that $j\cdot q$ is a common denominator for $\sem{\zeta'}{(\A,\overline
    F)}$, viewed as a function from $A^{\ell+m}$ to $\Qbot$.
    Then $j\cdot q$ is also a common denominator for $\sem{\zeta}{(\A,\overline
    F)}$, viewed as a function from $A^{\ell}$ to $\Qbot$. Hence
  we can let $d\coloneqq d'$.

  For $\vec b\in A^\ell$, let
  $\vec c_1,\ldots,\vec c_n$ be a list of all $\vec c\in A^m$ such that
  $(\A,\overline F)\models\varphi(\vec b,\vec c)$. By the induction
  hypothesis, for all $i\in[n]$ we have $|\sem{\zeta'}{(\A,\overline
    F)}(\vec b,\vec c_i)|\cdot j\cdot q\le c'\cdot p$. Hence
  \[
    |\sem{\zeta}{(\A,\overline
    F)}(\vec b)|\cdot j\cdot q\le\sum_{i=1}^n \big|\sem{\zeta'}{(\A,\overline
      F)}(\vec b,\vec c_i)\big|\cdot j\cdot q\le n\cdot c'\cdot p\le |A|^m\cdot c'\cdot p.
  \]
  We let $c\coloneqq |A|^m\cdot c'$.

  This completes the proof of the claim.

  \medskip Let $\A$ be a rational $\Upsilon$-structure.  To evaluate
  the term $\eta$ in \eqref{eq:2}, we compute the sequence of
  functions $F^{(t)}:A^k\to\Qbot$ for $t\in\{0,\ldots,|A|^k\}$. Recall
  that $F^{(0)}(\vec a)=\bot$ for all $\vec a$ and
  $F^{(t+1)}=\sem{\theta}{(\A,F^{(t)})}$. Choose $c,d$ according to
  Claim~1 applied to $\theta(\vec x)$. Then for all $t\in\Nat$, if $q$ is a common denominator
  for $F^{(t)}$ then for some $j\in[d]$, $jq$ is a common denominator
  for $F^{(t+1)}$. Moreover, if $p\in\PNat$ such that
  $p\ge \big|F^{(t)}(\vec a)\big|\cdot q$ for all
  $\vec a\in A^k$ with $F^{(t)}(\vec a)\neq\bot$, then
  $c\cdot p\ge \big|F^{(t+1)}(\vec a)\big|\cdot j\cdot q$
  for all $\vec a\in A^k$ with $F^{(t+1)}(\vec a)\neq\bot$.

  Observe that $1$ is a common denominator for $F^{(0)}$ and
  $1\ge |F^{(0)}(\vec a)|$ for all
  $\vec a\in A^k$ with $F^{(0)}(\vec a)\neq\bot$. An easy induction
  shows that for every $t\ge 1$ there is a $q_t\le d^t$ such that
  $q_t$ is a common denominator for $F^{(t)}$, and
  $c^t\ge |F^{(t)}(\vec a)|\cdot q_t$
  for all $\vec a\in A^k$ with $F^{(t)}(\vec a)\neq\bot$.

  Since $\|c\|,\|d\|= \|A\|^{O(1)}$, it follows that for $t\le |A|^k$
  it holds that $\sum_{\vec a\in A^{k}}\|F^{(t)}(\vec
  a)\|=\|A\|^{O(1)}$. Thus we can compute
  $F^{(t+1)}=\sem{\theta}{(\A,F^{(t)})}$ in polynomial time using Lemma~\ref{lem:eval-fo}.

\subsection{Proof of Theorem~\ref{theo:inexp}}
\label{app:theo:inexp}

Let $\A,\B$ be a $\Upsilon$-structures. An \emph{isomorphism} form
$\A$ to $\B$ is a bijective mapping $\pi:A\to B$ satisfying the
following two conditions.
\begin{enumerate}
  \item[(i)] For all $k$-ary relation symbols $R\in\Upsilon$ and all tuples
    $\vec a\in A^k$ it holds that $\vec a\in R^\A\iff\pi(\vec a)\in
    R^\B$.
  \item[(ii)] For all $k$-ary weight-function symbols $F\in\Upsilon$ and all tuples
    $\vec a\in A^k$ it holds that $F^\A(\vec a)=F^\B(\pi(\vec a))$.    
\end{enumerate}
A \emph{local isomorphism} from $\A$ to $\B$ is a bijection
$\lambda$ from a set $\dom(\lambda)\subseteq A$ to a set
$\rg(\lambda)\subseteq B$ satisfying conditions (i) and (ii) for
all tuples $\vec a\in\dom(\lambda)^k$. It will often be convenient to
describe local isomorphisms as sets $\lambda\subseteq A\times B$ of pairs.

Let $k,\ell\in\Nat$ such that $\ell\le k$ and $\vec a=(a_1,\ldots,a_\ell)\in A^\ell,\vec
b=(b_1,\ldots,b_\ell)\in B^\ell$.
The \emph{bijective $k$-pebble game} on $\A,\B$ with initial
position $\vec a,\vec b$ is played by two
players called \emph{Spoiler} and \emph{Duplicator}. A \emph{position}
of the game is a set $p\subseteq A\times B$ of size $|p|\le k$; the
\emph{initial position} is $p_0\coloneqq\{(a_i,b_i)\mid i\in[\ell]\}$. If
$|A|\neq|B|$ or if $p_0$ is not a local isomorphism, the game
ends immediately and Spoiler wins. Otherwise, a \emph{play} of the game proceeds in a possibly infinite sequence of \emph{rounds}. Each round consists
of the following steps (a)--(c).  Suppose the position before the round is $p$. 
\begin{enumerate}
\item[(a)] Spoiler selects
  a subset $p'\subseteq p$ of size $|p'|<k$.
\item[(b)] Duplicator selects a bijection $\beta:A\to B$. 
\item[(c)] Spoiler selects an $a\in A$, and the new position is $p\cup\{(a,\beta(a))\}$.
\end{enumerate}
If during the play a position $p$ that is not a local isomorphism is
reached, the play ends and Spoiler wins. Otherwise, the play
continues. If the play never ends, that is, each position is a local
isomorphism, Duplicator wins. 

We denote the game by $\BP_k(\A,\vec a,\B,\vec b)$, or just
$\BP_k(\A,\B)$ if $\ell=0$.

Observe that without loss of generality we may assume that in step (a)
of each round, if the current position $p$ has size $|p|<k$ then
Spoiler selects $p'=p$, and if $|p|=k$ then Spoiler selects a
$p'\subset p$ of size $|p'|=k-1$.

\begin{lemma}\label{lem:pebble-game}
  Let $k,\ell\in\Nat$ such that $\ell\le k$. Furthermore, let
  $\A,\B$ be $\Upsilon$-structures and $\vec a=(a_1,\ldots,a_\ell)\in A^\ell,\vec
b=(b_1,\ldots,b_\ell)\in B^\ell$ such that Duplicator has a
  winning strategy for the game $\BP_k(\A,\vec a,\B,\vec b)$. Then
  for all $\fosum$ formulas $\varphi(x_1,\ldots,x_\ell)$ with at most $k$
  variables it holds that
  \begin{equation}
    \label{eq:a1}
    \A\models\varphi(a_1,\ldots,a_\ell)\iff \B\models\varphi(b_1,\ldots,b_\ell),
  \end{equation}
  and for all $\fosum$ weight terms $\theta(x_1,\ldots,x_\ell)$ with at most $k$
  variables it holds that
  \begin{equation}
    \label{eq:a2}
    \theta^\A(a_1,\ldots,a_\ell)=\theta^\B(b_1,\ldots,b_\ell).
  \end{equation}
\end{lemma}

\begin{proof}
  The proof is by simultaneous induction on $\varphi$ and $\theta$. The
  base cases as well as the inductive steps for inequalities, Boolean
  connectives, arithmetic operators, and if-then-else are
  straightforward. The only interesting cases are quantification and
  summation.

  We consider quantification first.
  Assume $\varphi(x_1,\ldots,x_\ell)=\exists
  x\psi(x_1,\ldots,x_{\ell},x)$ and that
  $\A\models\varphi(a_1,\ldots,a_k)$. Let $a\in A$ such that
  $\A\models\psi(a_1,\ldots,a_k,a)$
  Consider the first round in
  the game $\BP_{k}(\A,\vec a,\B,\vec b)$. The initial position is
  $p_0=\{(a_i,b_i)\mid i\in[\ell]\}$. Suppose that in step (a), Spoiler
  selects the position $p'\coloneqq p_0$. This is possible, because
  $|p_0|\le\ell< k$. Let $\beta$ be the bijection
  selected by Duplicator in (b) according to her winning
  strategy. Suppose that in step (c), Spoiler selects $a$, and let $b\coloneqq\beta(a)$. Then the
  new position is $p_0\cup\{(a,b)\}$, and Duplicator wins the game
  $\BP_{k}(\A,\vec aa,\B,\vec bb)$. By the induction hypothesis,
  $\A\models\psi(a_1,\ldots,a_k,a)\iff
  \B\models\psi(b_1,\ldots,b_k,b)$ and thus
  $\B\models\psi(b_1,\ldots,b_k,b)$ by the choice of $a$. Thus
  $\B\models \varphi(b_1,\ldots,b_k)$.

  Similarly, if $\B\models \varphi(b_1,\ldots,b_k)$ then
  $\A\models\varphi(a_1,\ldots,a_k)$. This proves \eqref{eq:a1} for
  $\varphi(x_1,\ldots,x_\ell)=\exists
  x\psi(x_1,\ldots,x_{\ell},x)$.

  Formulas $\varphi(x_1,\ldots,x_\ell)=\forall
  x\psi(x_1,\ldots,x_{\ell},x)$ can be dealt with similarly.

  The most interesting case is that of summation terms. Consider
  such a term
  \[
    \theta(x_1,\ldots,x_\ell)=\sum_{(x_{\ell+1},\ldots,x_m):\varphi(x_1,\ldots,x_m)}\eta(x_1,\ldots,x_m),
  \]
  for some $m>\ell$, formula $\varphi(x_1,\ldots,x_m)$, and term
  $\eta(x_1,\ldots,x_m)$. 

  We consider the first $m-\ell$ rounds of the game
  $\BP_{k}(\A,\vec a,\B,\vec b)$ where Duplicator plays according to
  her winning strategy. Let
  $p_0=\{(a_i,b_i)\mid i\in[\ell]\}$ be the initial position. As
  $m\le k$, we can assume that in step (a) of each of the first
  $(m-\ell)$ rounds Spoiler just selects the current position of size
  $<k$. For every tuple
  $\vec a'=(a_{\ell+1},\ldots,a_m)\in A^{m-\ell}$ we define a sequence
  $\beta^{\vec a'}_1,\ldots,\beta^{\vec a'}_m$ and a a tuple
  $\vec b^{\vec a'}=(b^{\vec a'}_{\ell+1},\ldots, b^{\vec a'}_{m})\in
  B^{m-\ell}$ inductively as follows: $\beta^{\vec a'}_1$ is the
  bijection selected by Duplicator in the first round of the game in
  step (b), and $b_1^{\vec a'}\coloneqq\beta^{\vec
  a'}_1(a_{\ell+1})$. Assuming that Spoiler selects $a_{\ell+1}$ in
step (c), the new position $p_1\coloneqq p_0\cup\{(a_{\ell+1},b^{\vec
  a'}_1)\}$ is a winning position for Duplicator. For the inductive
step, consider some $i<m-\ell$ and suppose that the position 
$p_i=p_0\cup\{(a_{\ell+1},b^{\vec
  a'}_1),\ldots, (a_{\ell+i},b^{\vec
  a'}_i)\}$ is a winning position for Duplicator. Let $\beta^{\vec
  a'}_{i+1}$ be the
  bijection selected by Duplicator in the $(i+1)$st round of the game in
  step (b), and $b_{i+1}^{\vec a'}\coloneqq\beta^{\vec
  a}_{i+1}(a_{\ell+i+1})$. Assuming that Spoiler selects $a_{\ell+i+1}$ in
step (c), the new position $p_{i+1}\coloneqq p_{i+1}=p_0\cup\{(a_{\ell+1},b^{\vec
  a'}_1),\ldots, (a_{\ell+i+1},b^{\vec
  a'}_{i+1})\}$ is still a winning position for Duplicator.

Thus by the
induction hypothesis, we have 
\begin{gather}
  \label{eq:a3}
  \A\models
  \varphi(a_1,\ldots,a_m)\iff\B\models\varphi(b_1,\ldots,b_\ell,b_1^{\vec
    a'},\ldots,b_{m-\ell}^{\vec a'}),\\
  \label{eq:a4}
  \eta^\A(a_1,\ldots,a_m) = \eta^\B(b_1,\ldots,b_\ell,b_1^{\vec
    a'},\ldots,b_{m-\ell}^{\vec a'}).
\end{gather}
Let $\beta:A^{m-\ell}\to B^{m-\ell}$ be the mapping defined by
$\beta(\vec a')\coloneqq(b_1^{\vec a'},\ldots,b^{\vec a'}_{m-\ell})$. We shall prove that $\beta$ is bijective. Once we have
proved this, \eqref{eq:a3} and \eqref{eq:a4} imply $\theta^{\A}(\vec
a)=\theta^{\B}(\vec b)$.

To prove that $\beta$ is injective, consider distinct tuples $\vec
a'=(a'_{\ell+1},\ldots,a'_{m}), \vec
a''=(a''_{\ell+1},\ldots,a''_{m})\in A^{m-\ell}$. Let $i\in\{0,\ldots,m-\ell-1\}$ be such
that $a'_{\ell+j}=a''_{\ell+j}$ for all $j\le i$ and $a'_{\ell+i+1}\neq
a''_{\ell+i+1}$. Since $b_1^{\vec a'},\ldots,b_i^{\vec a'}$ and
$\beta_{i+1}^{\vec a'}$ only depend on $a'_{\ell+1},\ldots,a'_{\ell+i}$,
we have $\beta_{i+1}^{\vec a'} =\beta_{i+1}^{\vec a''}$. As
$\beta_{i+1}^{\vec a'}:A\to B$ is a bijection and $a'_{\ell+i+1}\neq
a''_{\ell+i+1}$, we have
\[
  b^{\vec a'}_{i+1}=\beta_{i+1}^{\vec a'}(a'_{\ell+i+1})\neq \beta_{i+1}^{\vec a'}(a''_{\ell+i+1})=\beta_{i+1}^{\vec
    a''}(a''_{\ell+i+1})=b^{\vec a''}_{i+1}.
\]
Thus $\beta(\vec a')\neq\beta(\vec a'')$, which proves that $\beta$ is
injective. Since the domain $A^{m-\ell}$ and co-domain $B^{m-\ell}$
of $\beta$ are finite sets of the same size, it follows that $\beta$
is bijective. 
\end{proof}

To prove Theorem~\ref{theo:inexp}, we apply the following well-known
result.

\begin{theorem}[Cai, Fürer and Immerman~\cite{cfi}]\label{theo:cfi}
  For every $k\in\Nat$ there are graphs $\A_k,\B_k$ such that
  $\A_k\not\cong B_k$ and Duplicator has a
  winning strategy for the game $\BP_k(\A_k,\B_k)$.

  Furthermore, $\A_k$ and $\B_k$ are distinguishable in polynomial
  time. That is, there is a polynomial time algorithm that accepts all
  $\A_k$ and rejects all $\B_k$.
\end{theorem}

\begin{proof}[Proof of Theorem~\ref{theo:inexp}]
  Let $Q$ be the Boolean query that is true on a structure $\A$ if
  the polynomial time algorithm of Theorem~\ref{theo:cfi} accepts $\A$ and false
  otherwise. Suppose for contradiction that there is an
  \ifpsum\ sentence $\varphi$ that expresses $Q$. By
  Lemma~\ref{lem:normal-form}, we may assume that $\varphi=\exists \vec
  x\Big(\chi(\vec x)\wedge\ifp\big(F(\vec x)\gets\theta\big)(\vec
  x)\Big)$ for some selection condition $\chi$ and \fosum\ term
  $\theta(\vec x)$. Let $\ell\coloneqq|\vec x|$, and let $m$ be the
  number of variables (free or bound) occurring in $\theta$.

  Let $k\coloneqq2\ell+m+1$ and consider the structures
  $\A\coloneqq\A_k$ and $\B\coloneqq\B_k$
  of Theorem~\ref{theo:cfi}. Then
  \begin{equation}
    \label{eq:a6}
    \A\models\varphi\quad\text{and}\quad \B\not\models\varphi.
  \end{equation}
  Furthermore, Duplicator has a
  winning strategy for the game $\BP_k(\A_k,\B_k)$, which implies that
  $\A$ and $\B$ have the same order, say, $n$. Furthermore, by
  Lemma~\ref{lem:pebble-game}, $\A$ and $\B$ satisfy the same
  \fosum\ sentences with at most $k-1$ variables.

  Consider the sequences $F^{(t)}_\A$ and $F^{(t)}_\B$, for
  $t\in\Nat$, the we compute when evaluating the term $\ifp\big(F(\vec x)\gets\theta\big)(\vec
  x)$ in $\A$ and $\B$, respectively. We claim that for every $t\ge 1$
  there is an \fosum\ term $\theta^{(t)}(\vec x)$ with at most
  $\ell+m$ variables such that $F^{(t)}_\A(\vec
  a)=\sem{\theta^{(t)}}{\A}(\vec a)$ for all $\vec a\in A^\ell$ and $F^{(t)}_\B(\vec
  b)=\sem{\theta^{(t)}}{\B}(\vec b)$ for all $\vec b\in B^\ell$. We
  let $\theta^{(1)}$ be the term obtained from $\theta$ by replacing
  each subterm $F(\vec y)$ in $\theta$ by the constant
  $\bot$. Furthermore, for every $t\ge 1$ we let $\theta^{(t+1)}$ be
  the term obtained from $\theta$ by replacing each subterm $F(\vec
  y)$ by
  \[
    \zeta(\vec y)\coloneqq\sum_{\vec z: \vec z=\vec y}\sum_{\vec x:\vec x=\vec
      z}\theta^{(t)}(\vec x).
  \]
  Here $\vec z$ is an $\ell$-tuple of variables disjoint from both
  $\vec y$ and $\vec x$, and $\vec z=\vec y$ abbreviates
  $\bigwedge_{i=1}^\ell z_i=y_i$. The role of the two summation
  operators is simply to put the right variables into the term
  $\theta^{(t)}$; for all structures $\C$ and tuples $\vec c\in
  C^\ell$ it holds that $\sem{\zeta(\vec y)}{\C}(\vec
  c)=\sem{\theta^{(t)}(\vec x)}{\C}(\vec c)$. Thus, semantically,
  $\theta^{(t+1)}$ is obtained from $\theta$ by replacing $F$ by a
  term defining $F^{(t)}$. Furthermore, note that the term
  $\theta^{(t+1)}$ contains only the variables in $\vec z$ and $\vec
  x$ in addition to those in $\theta$ and thus has at most $2\ell+m$ variables.

  Since the fixed-point process converges in at most $t\coloneqq n^\ell$ steps,
  we have $\sem{\ifp\big(F(\vec x)\gets\theta\big)(\vec
  x)}{\A}(\vec a)=\sem{\theta^{(t)}}{\A}(\vec a)$ for all $\vec a\in A^\ell$, and similarly $\sem{\ifp\big(F(\vec x)\gets\theta\big)(\vec
  x)}{\B}(\vec b)=\sem{\theta^{(t)}}{\B}(\vec b)$ for all $\vec b\in
B^\ell$. Thus 
\[
  \A\models\varphi\iff\A\models \exists \vec
  x\big(\chi(\vec x)\wedge\theta^{(t)}(\vec x)\big)\iff \B\models \exists \vec
  x\big(\chi(\vec x)\wedge\theta^{(t)}(\vec x)\big)
  \iff \B\models\varphi,
\]
because the formula $\exists \vec
  x\big(\chi(\vec x)\wedge\theta^{(t)}(\vec x)\big)$ has at most
  $2\ell+m$ variables. This contradicts \eqref{eq:a6}.
\end{proof}

\subsection{Proof of Theorem~\ref{theor:not-ptime}}

Let us refer to a boolean combination of polynomial inequalities in a
single variable $X$ simply as a \emph{condition}.  A
\emph{description} is an expression of the form
  \begin{tabbing}
    \textsf{if $\gamma_0$ then $\bot$} \\
    \textsf{else if $\gamma_1$ then $r_1$} \\
    \textsf{else if $\gamma_2$ then $r_2$} \\
    \dots \\
    \textsf{else $r_n$}
  \end{tabbing}
  where the $\gamma_i$ are conditions and the $r_i$ are rational
  functions in a single variable $X$.  We say that such a
  description $\delta$ is \emph{well-defined} if for every real number $a$ that
  does not satisfy $\gamma_0$, all the $r_i$ are well-defined
  on $a$ (i.e., no division by zero is performed).  For any $a
  \in \R$ the value $\delta(a) \in \Rbot$ is now defined in the
  obvious way.

Let $\K_1 \subseteq \K(1,1)$ denote the class of FNNs depicted in
Figure~\ref{fig:edge-split}(c), where all biases
are set to $0$.  Every network $\Net \in \K_1$ has one input node
and one output node, which we always denote by
$\inn$ and $\out$.  The network has any number of hidden
nodes; we denote this number by $H(\Net)$.
Note that any two hidden nodes are symmetric.

A \emph{symbolic assignment} is a mapping $\sigma$ from
a finite set $Y$ of variables to $\{\inn,\out,h\}$, where $h$ is a
symbol with the meaning of `hidden'.
Importantly, on any finite $Y$ there are only a finite number of symbolic
assignments.  Now an actual assignment $\nu$
on $Y$ in some $\Net \in \K_1$
is said to be \emph{of sort $\sigma$} if 
for each variable $y\in Y$, we have $\nu(y)=\inn$ iff
$\sigma(y)=\inn$;
$\nu(y)=\out$ iff $\sigma(y)=\out$; and $\nu(y)$ is a hidden node
iff $\sigma(y)=h$.

By an intricate but tedious induction, we can verify the
following.
\begin{lemma} \label{lem:description}
  For every \fosum\ formula $\varphi$ and
  every symbolic
  assignment $\sigma$ on the free variables of $\varphi$,
  there exists a condition $\gamma$ such that for every
  $\Net \in \K_1$ and every assignment $\nu$ in $\Net$ of sort
  $\sigma$, we have $\Net,\nu \models \varphi$ iff
  $\gamma(H(\Net))$ holds.

  Similarly, for every \fosum\ weight term $\theta$ and 
  every symbolic valuation $\sigma$ on the free variables of
  $\theta$, there exists a well-defined
    description $\delta$ such that for every
  $\Net \in \K_1$ and every assignment $\nu$ in $\Net$ of sort
  $\sigma$, we have $\sem\theta{(\Net,\nu)}=\delta(H(\Net))$.
\end{lemma}

Consider now the query $Q$ on $\K(1,1)$ where $Q(\Net)$ is true
iff $f^\Net(1)$ is an even natural number.
For any $\Net \in \K_1$, this means that $H(\Net)$ is even.
Suppose, for the sake of contradiction, that there exists a
closed \ifpsum\ formula $\varphi$ such that $\Net \models
\varphi$ iff $H(\Net)$ is even, for every $\Net \in \K_1$.  Since
all hidden nodes in such $\Net$ are symmetric, it is easy to see
that all fixpoints in $\varphi$ are reached in a constant number
of iterations.  So, without loss of generality, we may assume
$\varphi$ to be in \fosum.  By Lemma~\ref{lem:description} then,
noting that $\varphi$ has no free variables,
there exists a condition $\gamma$
such that $\gamma(a)$ holds iff $a$ is even,
for all natural numbers $a$.  This is impossible, since boolean
combinations of polynomial inequalities on $\R$ can only define
finite unions of intervals \cite{bcr_real}.

\subsection{Proof of Theorem \ref{theo:scalarptime}}

  For simplicity we give the proof for boolean queries without
  parameters.
  Let $\Net\in\K(m,p)$. Suppose that $\Net$ has $\pol$-bounded reduced
  weights for some polynomial $\pol(X)$. We choose a $c\in\Nat$ such
  that $\pol(n)< n^c$ for all $n\ge 2$.
  
As a first step of the proof, we observe that the nodes
of $\RNet$ can be linearly ordered in a canonical way.
Let $i_1,\dots,i_m$ and $o_1,\dots,o_p$ be the input and output
  nodes of $\Net$.
  We let
$\tilde i_1<\ldots<\tilde i_m<\tilde o_1<\ldots<\tilde o_m$ (recall
that we always assume input nodes and output nodes to be distinct). For hidden nodes $\tilde u$, we let $\tilde
i_j<\tilde u<\tilde o_k$. For distinct hidden nodes $\tilde u,\tilde
u'$, if the depth
of $u$ is smaller than the depth of $u'$, we let $\tilde u<\tilde
u'$. If $u$ and $u'$ have the same depth, but $\bias(u)\neq\bias(u')$,
we let $\tilde u<\tilde u'$ if and only if $\bias(u)<\bias(u')$. If
$\bias(u)=\bias(u')$, we consider all nodes
$\tilde v_1<\ldots<\tilde v_m$ of smaller depth. Then there is some
$i\in[m]$ such that $\wt(\tilde v_i,u)\neq \wt(\tilde v_i,u')$,
because otherwise we would have $u\sim u'$ and thus
$\tilde u=\tilde u'$. We choose the minimum $i$ such that
$\wt(\tilde v_i,\tilde u)\neq \wt(\tilde v_i,\tilde u')$ and let
$\tilde u<\tilde u'$ if and only if
$\wt(\tilde v_i,\tilde u)<\wt(\tilde v_i,\tilde u')$.

The linear order $\le$ on $\RNet$ induces a quasi-order $\preceq$ on
$\Net$: we let $u\preceq v$ if $\tilde u\le\tilde v$. Note that $u\sim
v$ if and only if $u\preceq v$ and $v\preceq u$.

This quasi-order $\preceq$ on $\Net$ is
\sifpsum-definable. We first construct a term
$\theta_{\textup{dep}}(x)$ such that $\sem{\theta_{\textup{dep}}}{\Net}(u)$ is the
depth of $u$ in $\Net$. Then we can easily construct an \sifpsum\ formula
$\varphi_{\preceq}(x,y)$ such that
$\Net\models\varphi_{\preceq}(u,v)\iff u\preceq v$. Moreover, we
construct a term $\theta_{\wt}(x,y)$ such that
$\sem{\theta_{\textup{wt}}}{\Net}(u,v)$ is the weight of the edge
$(\tilde u,\tilde v)$ in $\RNet$.
To unify the notation, we also let $\theta_{\textup{bias}}(x)\coloneqq
b(x)$ be the term defining the bias of a node.
This way, we have essentially
defined $\RNet$ within $\Net$.

Let $n\coloneqq |\RNet|$. Then $n$ is the length of the quasi-order
$\preceq$. Note that $n\ge 2$, because $\RNet$ has at least one input
node and one output node. 
Recall
that we chose $c$ such that $\pol(n)<n^c$ for the polynomial $\pol(X)$
bounding the weights.
Let $\preceq_c$ be the lexicographical order on $c$-tuples associates
with $\preceq$. That is, for tuples
$\vec u=(u_1,\ldots,u_c),\vec v=(v_1,\ldots,v_c)$ we have
$\vec u\preceq_c\vec v$ if and only if either
$u_i\sim v_i$ for all $i\in[c]$ or for the minimal $i$ such that
$u_i\not\sim v_i$ it holds that $u_i\prec v_i$. We index positions in
the quasi-order $\preceq_c$
with numbers $i\in\{0,\ldots,n^c-1\}$. For each such $i$, we let
$\class i$
denote the $i$th equivalence class with respect to $\preceq_c$, and
for every $c$-tuple $\vec u$ of nodes of $\Net$ we let $\ind{\vec u}$
be the unique $i\in\{0,1\ldots,n^c-1\}$ such that $\vec u\in\class i$. 
We construct an \sifpsum\ formula
$\varphi_{\textup{lex}}(\vec x,\vec y)$ that defines $\preceq_c$ and a
term $\theta_{\textup{ind}}(\vec x)$ such that for every
$c$-tuple $\vec u$ we have
$\sem{\theta_{\textup{ind}}}{\Net}(\vec u)=\ind{\vec u}$. It
will also be convenient to let $\varphi_{\textup{slex}}(\vec x,\vec
y)\coloneqq \varphi_{\textup{lex}}(\vec x,\vec y)\wedge\neg
\varphi_{\textup{lex}}(\vec y,\vec x)$ be the formula that defines the
strict lexicographical order $\prec_c$.

As $\RNet$ has $\pol$-bounded 
  weights, every bias and weight of $\RNet$ can be written as a
  fraction $\frac{r}{q}$ where $|r|,q\le \pol(n)< n^c$. Suppose that
  for node $u$ of $\Net$ we have $\bias(\tilde
  u)=\bias(u)=\frac{r(u)}{q(u)}$ in reduced form and for every edge
  $(u,v)$ of $\Net$ we have $\wt(\tilde u,\tilde
  v)=\frac{r(u,v)}{q(u,v)}$ in reduced form. The next step may be the
  crucial step of the proof. We would like to define the numbers
  $r(u),q(u),r(u,v),q(u,v)$ in \sifpsum. However, it is not obvious
  how to do this directly by terms $\theta(x)$ or $\theta(x,y)$. We
  sidestep this issue by defining the index of the numbers in the
  quasi-order $\preceq_c$.
  We construct \sifpsum\
  formulas $\varphi_{\textup{bias}}(x,\vec z,\vec z')$ and
  $\varphi_{\textup{wt}}(x,y,\vec z,\vec z')$ such that for all nodes
  $u,v$ and $c$-tuples $\vec t,\vec t'$ of nodes
  of $\Net$ we have
  \begin{align*}
    \Net\models \varphi_{\textup{bias}}(u,\vec t,\vec t')&\iff \ind{\vec
                                                          t}=|r(u)|\text{
                                                          and }\ind{\vec
                                                          t'}=q(u),\\
    \Net\models \varphi_{\textup{wt}}(u,v,\vec t,\vec t')&\iff \ind{\vec
                                                          t}=|r(u,v)|\text{
                                                          and }\ind{\vec
                                                          t'}=q(u,v).
  \end{align*}
  To construct $\varphi_{\textup{bias}}(x,\vec z,\vec z')$,
  we first take care of the sign, letting
  \[
    \varphi_{\textup{bias}}(x,\vec
    z,\vec z')\coloneqq\big(\theta_{\textup{bias}}(x)<0\wedge\varphi_{\textup{neg}}(x,\vec
    z,\vec z'\big)\vee
    \big(\theta_{\textup{bias}}(x)\ge
    0\wedge\varphi_{\textup{pos}}(x,\vec z,\vec z'\big).
  \]
  Now we let
  \[
    \varphi_{\textup{pos}}(x,\vec z,\vec z'\big)\coloneqq
\big(\theta_{\textup{bias}}(x)\cdot \theta_{\textup{ind}}(\vec
z')=\theta_{\textup{ind}}(\vec z)\big)\wedge\forall \vec
y'\Big(\varphi_{\textup{slex}}(\vec y',\vec z')\to\neg\exists\vec y\, \theta_{\textup{bias}}(x)\cdot \theta_{\textup{ind}}(\vec
y')=\theta_{\textup{ind}}(\vec y)\Big).
\]
To define $\varphi_{\textup{neg}}$, we simply replace both occurrences of
$\theta_{\textup{bias}}(x)$ in $\varphi_{\textup{pos}}$ by $(-1)\cdot
\theta_{\textup{bias}}(x)$. The formula $\varphi_{\textup{wt}}(x,y,\vec
z,\vec z')$  can be defined similarly.

Using all these \sifpsum\ expressions, we can define an ordered copy
of $\RNet$ in $\Net$; formally, this is done by a transduction
\cite{grohe_book}.  By the Immerman-Vardi
Theorem~\cite{imm_relpol,vardi_comp}, every polynomial-time computable
query on ordered structures is expressible in \IFP\ and hence in
$\sifpsum$. Thus we obtain an $\ifpsum$ formula that defines, in $\Net$, the
answer to query $Q$ applied to $\RNet$. As the query is
model-agnostic, this also gives us the answer to $Q$ applied to
$\Net$.

\subsection{Proof of Theorem \ref{thm:np-hardness-fx-zero}}

	Let $\varphi\coloneqq\bigwedge_{i=1}^m(\lambda_{i1}\wedge\lambda_{i2}\wedge\lambda_{i3})$,
	where $\lambda_{ij}\in\{X_k,\neg X_k\}$ for some $k\in[n]$, be a 3-CNF
	formula in the Boolean variables $X_1,\ldots,X_n$. 
	In the following, we will construct an FNN $\Net$ such that $\f^\Net$ is the zero function if and only if $\varphi$ is unsatisfiable. The construction is based on the idea of interpreting numbers
        with binary representation $(0.a_1a_2\ldots a_n)_2$ as variable assignments and letting $\Net$ simulate $\varphi$ on these numbers. 
	To that end, we first use an auxiliary FNN $\Netsplit \in \K(1,n)$ with the following properties. For all $i \in [n]$, we have
	\begin{enumerate}
		\item\label{itm:in-zero-one} $\forall x\, \f^{\Netsplit}_{\out_i}(x) \in [0,1]$ and
		\item\label{itm:agree-on-bin}  if $x \in [0,1)$ is a
                  multiple of $2^{-n}$ with $x = (0.a_1a_2\ldots
                  a_n)_2$, then $\f^{\Netsplit}_{\out_i}(x) =
                  a_i$.
	\end{enumerate} 
	It is well-known how to construct such a network $\Netsplit$. For the reader's convenience we give the construction at the end of the proof. Next, we aim to simulate $\varphi$ with an FNN. For that, let $f_\varphi : \R^n \to [0,1]$ be defined by
	\[
	\f_\varphi(\vec x)\coloneqq
	\min_{i\in[m]}\max_{j\in[3]}\ell_{ij}(\vec x)
	\]
	where for $\vec x=(x_1,\ldots,x_n)$ we let $\ell_{ij}(\vec x)=x_k$ if $\lambda_{ij}=X_k$ and
	$\ell_{ij}(\vec x)=1-x_k $ if $\lambda_{ij}=\neg X_k$.
	Then
	for all $\vec x=(x_1,\ldots,x_n)\in\{0,1\}^n$ we have $f_\varphi(\vec
	x)\in\{0,1\}$ with $f_\varphi(\vec x)=1$ if and only if the assignment
	$X_k\mapsto x_k$ satisfies $\varphi$. Observe that one can
    easily construct an FNN $\Net_\varphi\in\K(n,1)$ with $\f^{\Net_\varphi} = f_\varphi$ since $\max(x_i, x_j) = x_i + \ReLU(x_j - x_i)$.
	It is tempting to consider an FNN that computes $f^{\Net_\varphi} \circ f^{\Netsplit}$ as a candidate for $\Net$. Indeed, if $\varphi$ is satisfiable, then the satisfying assignment $X_i \mapsto a_i$ yields an input $x=(0.a_1a_2\ldots a_n)_2$ with $f^{\Net_\varphi} \big( f^{\Netsplit}(x)\big) = 1$. If $\varphi$ is unsatisfiable, however, we want $f^\Net$ to be $0$ for all $x \in \R$ and not just for multiples of $2^{-n}$. Luckily, we can ensure this requirement using the following insight.
	
	\medskip
	\textsc{Claim~2. }
	{\itshape 
	Suppose that $\varphi$ is unsatisfiable. Then $0\le f_\varphi(\vec x)\le 1/2$
	for all $\vec x\in[0,1]^n$.
	}
	
	\medskip
	\textsc{Proof of Claim~2:}
	Let $c\coloneqq\max_{\vec x\in[0,1]^n}f_\varphi(\vec x)$.
We need to prove that $c\le 1/2$.
        Assume that
		$c>0$. Choose $\vec x=(x_1,\ldots,x_n)\in[0,1]^n$ such that
		$f_\varphi(\vec x)=c$ with the minimum number $k\in[n]$ such that
		$x_k\not\in\{0,1\}$.
                Then there is an $i\in[m]$ such that
		$X_k\in\{\lambda_{i1},\lambda_{i2},\lambda_{i3}\}$ and
		$x_k=\max_{j\in[3]}\ell_{ij}(\vec x)$ because otherwise we
		could set $x_k$ to $0$ without decreasing $f_\varphi(\vec x)$. Similarly,
		there is an $i'\in[m]$ such that
		$\neg X_k\in\{\lambda_{i'1},\lambda_{i'2},\lambda_{i'3}\}$ and
		$1-x_k=\max_{j\in[3]}\ell_{i'j}(\vec x)$, because otherwise we
		could set $x_k$ to $1$ without decreasing $f_\varphi(\vec x)$.
		
		As either $x_k\le 1/2$ or $1-x_k\le 1/2$, either
		$\max_{j\in[3]}\ell_{ij}(\vec x)\le 1/2$ or
		$\max_{j\in[3]}\ell_{i'j}(\vec x)\le 1/2$ and thus
		$c=f_\varphi(\vec x)\le 1/2$.
	This proves the claim.
	
	\medskip
        Since $\lbrace \f^{\Netsplit}(x) \,\mid\, x \in \R \rbrace \subseteq [0,1]^n$, if  $\varphi$ is  unsatisfiable, we get $\max_{x \in \R} f^{\Net_\varphi} \big( f^{\Netsplit}(x)\big) \leq \frac12$. Therefore, we let $\Net$ consist of the concatenation of $\Netsplit$ and $\Net_\varphi$ and connect it with an edge of weight $2$ to a neuron with bias $-1$. Then $\Net$ computes the function 
	\[
	\f^\Net(x) = \ReLU\Big(- 1 + 2\cdot   \f^{\Net_\varphi} \big( \f^{\Netsplit} (x)\big)\Big)
	\]
	which is constantly zero for unsatisfiable $\varphi$ and reaches $1$ otherwise. 
	
	\bigskip
	It remains to present the construction of $\Netsplit$. For convenience, we will
        use the \emph{linearised sigmoid function} (also known as $\ReLU1$)
        defined by
	\[
	\clipped(x) = 	\ReLU(x) - \ReLU(x-1) = \begin{cases}
		1  & \text{if } x > 1\text,\\
		x & \text{if } 0 \leq x \leq  1\text,\\
		0  & \text{if } x < 0\text.\\
	\end{cases}
	\]
	We first aim to construct an FNN $\Net_1$ in $\K(1,1)$ with range $[0,1]$ and $\f^{\Net_1}((0.a_1a_2\ldots a_n)_2) = a_1$. To that end, we observe
	\[2 \cdot (0.a_1a_2\ldots a_n)_2 - 1
	\begin{cases}
		\geq 0  & \text{if } a_1 = 1\text,\\
		\leq -2^{-(n-1)}  & \text{if } a_1 = 0\text.\\
	\end{cases}
	\] 
	Now,  we can amplify this gap and extract $a_1$ via
	\[
	a_1 = \clipped(
	(2 \cdot (0.a_1a_2\ldots a_n)_2-1)\cdot 2^{n-1} + 1
	).
	\]
    Because of the previous equation, we choose $\Net_1$ to compute $\f^{\Net_1}(x) = \clipped(
	(2 x-1)\cdot 2^{n-1} + 1
	)$ and obtain the desired properties. To lift this to a construction of $\Netsplit$, we observe $(2
    \cdot (0.a_1a_2\ldots a_n)_2-a_1) = (0.a_2a_3\ldots a_n)$. Therefore, we can apply this construction iteratively and obtain the FNN $\Netsplit$ which fulfills the properties by computing 
    \[
    \f^{\Netsplit}_{\out_i}(x) = \clipped\Big(
	\big(2^i x - 1 - \sum_{j=1}^{i-1} 2^{i-j}\f^{\Netsplit}_{\out_j}(x)\big)\cdot 2^{n-i} + 1
	\Big).
    \]

\end{document}